\documentclass[envcountsame,orivec,runningheads,a4paper
]{llncs}

\newcommand{\version}{1}
\newcommand{\commentvar}{0}  

\newcommand{\SLV}[2]{\ifthenelse{\equal{\version}{0}}{#1}{ #2}}

\usepackage[utf8]{inputenc}

\usepackage{microtype}
\usepackage{marginnote}
\usepackage{mathtools,amsthm}
\usepackage{multirow}
\usepackage[hidelinks,bookmarks,bookmarksopen,bookmarksdepth=2]{hyperref}

\hypersetup{colorlinks=false}	

\newtheorem{lemmaAppendix}{Lemma} 
\newtheorem{propositionAppendix}{Proposition} 

\usepackage{ifthen}
 \usepackage{xspace}
 
\newboolean{talk}
\setboolean{talk}{false}
\newboolean{paper}
\setboolean{paper}{false}

\newboolean{IEEEstyle}
\setboolean{IEEEstyle}{false}
\newboolean{lipicsstyle}
\setboolean{lipicsstyle}{false}
\newboolean{eptcsstyle}
\setboolean{eptcsstyle}{false}
\newboolean{entcsstyle}
\setboolean{entcsstyle}{false}
\newboolean{sigplanstyle}
\setboolean{sigplanstyle}{false}
\newboolean{easychairstyle}
\setboolean{easychairstyle}{false}
\newboolean{scrartclstyle}
\setboolean{scrartclstyle}{false}
\newboolean{lncsstyle}
\setboolean{lncsstyle}{false}

\newboolean{needstheorems}
\setboolean{needstheorems}{false}
\newboolean{withimages}
\setboolean{withimages}{false}
\newboolean{withproofs}
\setboolean{withproofs}{true}

\newboolean{french}
\setboolean{french}{false}

\setboolean{lncsstyle}{true}
\ifthenelse{\boolean{talk}}{}{\usepackage[utf8]{inputenc}}
\ifthenelse{\boolean{french}}{\usepackage[french]{babel}}{
  \ifthenelse{\boolean{lipicsstyle}}{}{\usepackage[english]{babel}}}
\usepackage{amssymb}
\usepackage{amsmath}
\usepackage{graphicx}
\usepackage{bussproofs}
\usepackage{cmll}
\usepackage{stmaryrd}
 \usepackage{multirow}
 \ifthenelse{\boolean{talk}}{\usepackage{color}}{
   \ifthenelse{\boolean{lipicsstyle}}{\usepackage{color}}{\usepackage[usenames]{xcolor}}}
 \ifthenelse{\boolean{IEEEstyle}}{
 \usepackage[bookmarks={false}]{hyperref}}{
 \usepackage{hyperref}}
 \usepackage{fancybox}
 \usepackage{hhline}
\usepackage{nameref}
\usepackage{complexity}
\usepackage{appendix}
\usepackage{xcolor}
\usepackage[normalem]{ulem}

\ifthenelse{\boolean{IEEEstyle}}{
\setboolean{needstheorems}{true}}{}

\ifthenelse{\boolean{eptcsstyle}}{
\setboolean{needstheorems}{true}}{}

\ifthenelse{\boolean{scrartclstyle}}{
\setboolean{needstheorems}{true}}{}

\ifthenelse{\boolean{sigplanstyle}}{
\setboolean{needstheorems}{true}}{}

\ifthenelse{\boolean{easychairstyle}}{
\setboolean{needstheorems}{true}}{}

\ifthenelse{\boolean{lipicsstyle}}{
    }{}

\ifthenelse{\boolean{needstheorems}}{

\usepackage{amsthm}

    \newtheorem{theorem}{Theorem}[section]
    \newtheorem{lemma}[theorem]{Lemma}
    \newtheorem{corollary}[theorem]{Corollary}
    \newtheorem{proposition}[theorem]{Proposition}
    \newtheorem{definition}[theorem]{Definition}

    \ifthenelse{\boolean{scrartclstyle}}{
  \newtheorem{example}{Example}[section]
}{}

}{}

\newcommand{\myproof}[1]{
\ifthenelse{\boolean{withproofs}}{#1}{}
}

\newcommand{\lterms}{\Lambda}

\newcommand{\la}[1]{\lambda #1.}

\newcommand{\tm}{t}
\newcommand{\tmtwo}{s}
\newcommand{\tmthree}{u}
\newcommand{\tmfour}{r}
\newcommand{\tmfive}{p}
\newcommand{\tmsix}{q}

\newcommand{\var}{x}
\newcommand{\vartwo}{y}
\newcommand{\varthree}{z}

\newcommand{\Rew}[1]{\rightarrow_{#1}}

\renewcommand{\to}{\Rew{}}
\newcommand{\tob}{\Rew{\beta}}

\newcommand{\toone}{\Rew{1}}
\newcommand{\totwo}{\Rew{2}}

\newcommand{\lRew}[1]{\; \mbox{}_{#1}{\leftarrow}\ }

\newcommand{\lpartob}{\; \mbox{}_{\beta}{\Leftarrow}\ }

\newcommand{\parRew}[1]{\Rightarrow_{#1}}

\newcommand{\parto}{\parRew{}}

\newcommand{\partob}{\parRew{\beta}}

\newcommand{\toh}{\Rew{h}}
\newcommand{\towh}{\Rew{wh}}
\newcommand{\tonoth}{\Rew{\neg h}}

\newcommand{\tostrat}{\Rew{x}}

\newcommand{\esym}{{\mathtt e}}

\newcommand{\wsym}{w}

\newcommand{\cbv}{{CbV}\xspace}

\newcommand{\val}{v}

\newcommand{\wal}{w}
\newcommand{\waltwo}{\wal'}

\makeatletter
\newsavebox{\@brx}
\newcommand{\llangle}[1][]{\savebox{\@brx}{\(\m@th{#1\langle}\)}%
  \mathopen{\copy\@brx\kern-0.7\wd\@brx\usebox{\@brx}}}
\newcommand{\rrangle}[1][]{\savebox{\@brx}{\(\m@th{#1\rangle}\)}%
  \mathclose{\copy\@brx\kern-0.7\wd\@brx\usebox{\@brx}}}
\makeatother

\newcommand{\nbvctxtwo}[1]{\nbvctxtwo{#1}}

\newcommand{\eqdef}{:=}
\newcommand{\grameq}{::=}

\newcommand{\isub}[2]{\{#1/#2\}}

\renewcommand{\isub}[2]{\{#1{\shortleftarrow}#2\}}

\newcommand{\llbrace}{\{ \kern -0.27em \vert}
\newcommand{\rrbrace}{\vert \kern -0.27em \}}

\renewcommand{\l}{\lambda}
\newcommand{\ie}{{\em i.e.}\xspace}

\newcommand{\ih}{{\textit{i.h.}}\xspace}
\newcommand{\fv}[1]{\mathtt{fv}(#1)}

\ifthenelse{\boolean{entcsstyle}}{}{
	}

\newcommand{\ignore}[1]{}

\newcommand{\colspace}{@{\hspace{.5cm}}}
\newcommand{\myinput}[1]{\ifthenelse{\boolean{withimages}}{\input{#1}}{}}

\newcommand{\reflemma}[1]{Lemma~\ref{l:#1}}

\newcommand{\refthm}[1]{Theorem~\ref{thm:#1}}

\newcommand{\refprop}[1]{Proposition~\ref{prop:#1}}
\newcommand{\refpropp}[2]{Prop.~\ref{prop:#1}.\ref{p:#1-#2}} 
\newcommand{\refsect}[1]{Sect.~\ref{sect:#1}}

\ifthenelse{\boolean{easychairstyle}}{}{
  \ifthenelse{\boolean{lncsstyle}}{}{	
	
  }
}
\newcommand{\reffig}[1]{Fig.~\ref{fig:#1}}

\newcommand{\refdef}[1]{Def.~\ref{def:#1}}

\newcommand{\refpoint}[1]{Point~\ref{p:#1}}
\newcommand{\refpoints}[2]{Points~\ref{p:#1}-\ref{p:#2}}

\newcommand{\mellies}{{Melli{\`e}s}\xspace}

\newcommand{\rios}{{R{\'i}os}\xspace}

\newcommand{\set}[1]{\{#1\}}

\newcommand{\nat}{\mathbb{N}}

\newcommand{\toe}{\Rew{\esym}}
\newcommand{\toi}{\Rew{\isym}}
\renewcommand{\toi}{\Rew{\intsym}}
\newcommand{\tow}{\Rew{\wsym}}

\newcommand{\betav}{\beta_\val}

\newcommand{\tobv}{\Rew{\betav}}

\newcommand{\sizep}[2]{|#1|_{#2}}

\newcommand{\lo}{$\losym$\xspace}
\newcommand{\losym}{\ell o}
\newcommand{\LO}{$\losym$\xspace}

\newcommand{\sep}{\hspace*{0.35cm}}

\newcommand{\tolo}{\Rew{\losym}}

\newcommand{\withproofs}[1]{\ifthenelse{\boolean{withproofs}}{#1}{}}

\newcommand{\withoutproofs}[1]{\ifthenelse{\boolean{withproofs}}{}{#1}}

\ifthenelse{\boolean{entcsstyle}}{}{
  \ifthenelse{\boolean{lncsstyle}}{}{	
    
  }
 }

\newcounter{numberone}

\newcommand{\terms}{S}
\newcommand{\rewsys}{\mathcal S}
\newcommand{\extsym}{\esym}
\newcommand{\intsym}{{\lnot\extsym}}

\newcommand{\partoint}{\parRew{\neg\esym}}

\newcommand{\tonotlo}{\Rew{\neg \losym}}

\newcommand{\partoind}[1]{\overset{#1}{\Rightarrow}}
\newcommand{\partobind}[1]{\overset{#1}{\Rightarrow}_{\beta}}
\newcommand{\partobindlong}[1]{\overset{#1}{\parRew{\beta}}}

\newcommand{\partonotlo}{\parRew{\neg \losym}}
\newcommand{\partonoth}{\parRew{\neg h}}

\newcommand{\partobvind}[1]{\overset{#1}{\Rightarrow}_{\betav}}
\newcommand{\partobvindlong}[1]{\overset{#1}{\parRew{\betav}}}
\newcommand{\partobv}{\parRew{\betav}}
\newcommand{\tonotw}{\Rew{\neg \wsym}}
\newcommand{\partonotw}{\parRew{\neg \wsym}}

\newcommand{\essential}{essential\xspace}
\newcommand{\Essential}{Essential\xspace}

\newcommand{\external}{\essential}
\newcommand{\External}{\Essential}

\newcommand{\tmu}{\tmthree}
\newcommand{\tms}{\tmtwo}
\newcommand{\tmr}{\tmfour}
\newcommand{\tmq}{\tmsix}

\newcommand{\macrostep}{macro-step}

\newcommand{\strategy}{reduction\xspace}
\newcommand{\strategies}{reductions\xspace}

\renewcommand{\ll}{$\llsym$\xspace}
\newcommand{\llsym}{\ell\ell}

\newcommand{\ltow}{\lRew{\wsym}}

\newcommand{\complete}{full\xspace}

\newcommand{\completeness}{fullness\xspace}
\newcommand{\Completeness}{Fullness\xspace}
\newcommand{\sequence}{sequence\xspace}
\usepackage{tikz}
\usetikzlibrary{calc}
\usetikzlibrary{matrix,arrows}
\usetikzlibrary{decorations.shapes}
\usetikzlibrary{decorations.text}
\usetikzlibrary{decorations.pathmorphing}
\usetikzlibrary{decorations.markings}
\usetikzlibrary{positioning}

\tikzset{
node distance=1.1cm, auto,
every node/.style={font=\scriptsize },
ocenter/.style={baseline={([yshift=-.5ex, xshift=-.5ex]current bounding box)}},  
labelBeginAbove/.style={postaction={decorate,decoration={markings,mark=at position 0 with {\node[inner sep= 0.6pt, above=1pt]{\tiny #1};}} } },
labelBeginBelow/.style={postaction={decorate,decoration={markings,mark=at position 0 with {\node[inner sep= 0.6pt, below=1pt]{\tiny #1};}}}},
labelEndAbove/.style={postaction={decorate,decoration={markings,mark=at position 1 with {\node[inner sep= 0.6pt, above=1pt]{\tiny #1};}}}},
labelEndBelow/.style={postaction={decorate,decoration={markings,mark=at position 1 with {\node[inner sep= 0.6pt, below=1pt]{\tiny #1};}}}},
labelEndRight/.style={postaction={decorate,decoration={markings,mark=at position 1 with {\node[inner sep= 0.6pt, right=1pt]{\tiny #1};}}}},
labelEndLeft/.style={postaction={decorate,decoration={markings,mark=at position 1 with {\node[inner sep= 0.6pt, left=1pt]{\tiny #1};}}}}
}

\setboolean{withproofs}{false}
\setboolean{withimages}{true}

\newcommand{\commCF}[1]{\ifthenelse{\equal{\commentvar}{1}}{\textcolor{violet}{*CF: #1 *}}{}}

\newcommand{\toll}{\Rew{\llsym}}
\newcommand{\ltoll}{\lRew{\llsym}}

\renewcommand{\tostrat}{\toll}
\newcommand{\tonotll}{\Rew{\neg \llsym}}
\newcommand{\tonotstrat}{\tonotll}
\newcommand{\tostratind}[1]{\Rew{\beta:{#1}}}
\newcommand{\ltostratind}[1]{\lRew{\beta:{#1}}}
\newcommand{\partostratind}[1]{\Rightarrow_{\beta:{#1}}}
\newcommand{\partonotll}{\parRew{\lnot\llsym}}
\newcommand{\partonotstrat}{\partonotll}
\newcommand{\Deg}[1]{\llsym(#1)}
\renewcommand{\deg}[1]{\Deg{#1}}
\newcommand{\towb}{\rightarrow_{w\beta}}
\newcommand{\tonotwh}{\rightarrow_{\lnot wh}}

\newcommand{\tost}{\tostrat}

\newcommand{\partoat}{\partostratind}

\newcommand{\partonotst}{\partonotstrat}

\newcommand{\NoteProof}[1]{
	\marginnote{
		\scriptsize{Proof p.\,{\pageref{#1}}}}}
\newcommand{\NoteState}[1]{
	\marginnote{
		\scriptsize{See p.\,{\pageref{#1}}}}}

\author{Beniamino Accattoli\inst{1} 
	\and Claudia Faggian\inst{2}  
	\and Giulio Guerrieri\inst{3}}
\authorrunning{B. Accattoli et al.}
\institute{Inria \& LIX, École Polytechnique, UMR 7161, Palaiseau, France \email{\href{mailto:beniamino.accattoli@inria.fr}{beniamino.accattoli@inria.fr}}
	\and Université de Paris, IRIF, CNRS, F-75013 Paris, France  
\email{\href{mailto:faggian@irif.fr}{faggian@irif.fr}}
	\and University of Bath, Department of Computer Science, Bath, UK
\email{\href{mailto:g.guerrieri@bath.ac.uk}{g.guerrieri@bath.ac.uk}}
}

\title{Factorization and Normalization, Essentially}

\begin{document}

\maketitle

\begin{abstract}
$\l$-calculi come with no fixed evaluation strategy. Different strategies may then be considered, and it is important that they satisfy some abstract rewriting property, such as factorization or normalization theorems. 
 In this paper we provide simple proof techniques for these theorems. Our starting point is a revisitation of Takahashi's technique to prove factorization for  head reduction. Our technique is both simpler and more powerful, as it works in cases where Takahashi's does not. We then pair factorization with two other abstract properties, defining \emph{essential systems}, and show that normalization follows. Concretely, we apply the technique to four case studies, two classic ones, head and the leftmost-outermost reductions, and two less classic ones, non-deterministic weak call-by-value and least-level reductions.

\end{abstract}


\section{Introduction}
\label{sect:intro}

The $\l$-calculus is the model underlying functional programming languages and proof assistants. The gap between the model and its incarnations is huge. In particular, the $\l$-calculus does not come with a fixed reduction strategy, while concrete frameworks need one. 
A desirable property is that  the reduction which  is implemented terminates on all terms on which $\beta$ reduction has a reduction sequence to normal form. This is guaranteed by a \emph{normalization theorem}. 
 Two classic
 examples are  the \emph{leftmost-outermost} and \emph{head} normalization theorems (theorems 13.2.2 and  11.4.8 in Barendregt 
 \cite{Barendregt84}). The former states that if the term has a $\beta$-normal form, leftmost-outermost reduction is guaranteed to find it; the latter has a similar but 
 subtler statement, roughly head reduction computes a head normal form, if the term has any. 
 
 Another classic theorem for head reduction states that head \strategy approximates the $\beta$-normal form by computing an essential part of every evaluation sequence. The precise formulation is 
 a \emph{factorization theorem}:  a  sequence of $\beta$ steps $\tm\tob^*\tmtwo$ can always be re-arranged as a  sequence of head steps ($\toh$) followed by a sequence of non-head steps ($\tonoth$), that is, $\tm \toh^* \tmthree \tonoth^* \tmtwo$. 
 Both head and leftmost-outermost reductions play a key role in the theory of the $\l$-calculus as presented in 
 Barendregt \cite{Barendregt84} or Krivine \cite{krivine1993lambda}.

Variants of the $\l$-calculus abound and  are 
continuously introduced: weak, call-by-value, call-by-need, classical, with pattern matching, sharing, non-determinism, probabilistic choice, quantum features, differentiation, \textit{etc}. 
So, normalization and factorization theorems need to be studied in many variations. Concepts and techniques to prove these theorems do exist, but they do not have the essential, intuitive structure of other fundamental properties, such as~confluence.

 \paragraph{This paper.} Here we provide a presentation of factorization and 
 normalization revisiting a simple technique 
 due to Takahashi \cite{DBLP:journals/iandc/Takahashi95}, making it even \emph{simpler} and \emph{more widely applicable}.
 We separate the abstract reasoning from the concrete details of head reduction, 
 and apply the revisited proof method to several case studies.
 The presentation is novel and hopefully accessible to anyone familiar with the $\l$-calculus, without  a background in advanced  notions of rewriting theory.

We provide four case studies, 
all following the same method. 
Two are revisitations of the classic cases of head and leftmost-outermost (shortened to \lo) reductions. Two are folklore cases. The first is \emph{weak} (\ie out of $\l$-abstractions) \emph{call-by-value} (shortened to \cbv) \emph{reduction}
in its non-deterministic presentation. 
The second is \emph{least-level} (shortened to \ll) \emph{reduction}, a \strategy coming from the linear logic literature---sometimes called \emph{by levels}---and which is usually presented 
using proof nets (see de Carvalho, Pagani and Tortora de Falco \cite{DBLP:journals/tcs/CarvalhoPF11} or Pagani and Tranquilli \cite{DBLP:journals/mscs/PaganiT17}) or calculi related to proof nets (see Terui \cite{DBLP:journals/aml/Terui07} or Accattoli \cite{DBLP:conf/rta/Accattoli12}), rather than in the ordinary $\l$-calculus.
The \lo and \ll cases are \emph{\complete} reductions for $\beta$, 
\ie they have the same normal forms as $\beta$. 
The head and weak \cbv cases are not \complete, as they may not compute $\beta$  normal forms. 

\paragraph{Takahashi.} 
In \cite{DBLP:journals/iandc/Takahashi95}, Takahashi uses the natural inductive notion of \emph{parallel}\footnotemark
\footnotetext{ 
	The terminology at work in the literature on $\l$-calculus and the rewriting terminology often clash: the former calls \emph{parallel $\beta$ reduction} what the latter calls \emph{multi-step $\beta$ reduction}---parallel reduction in rewriting is something else.} 
$\beta$ reduction (which reduces simultaneously a number of $\beta$-redexes; it is also the key concept in Tait and Martin-L\"of's classic proof of confluence of the $\l$-calculus)
to introduce a simple proof technique 
for head factorization, from which head normalization follows. 
By iterating head factorization, she also obtains leftmost-outermost normalization, via a simple argument on the structure of~terms due to Mitschke \cite{Mitschke79}. 

Her technique has 
been employed 
for various $\l$-calculi because of its simplicity. 
Namely, for the $\l$-calculus with $\eta$ by Ishii \cite{DBLP:journals/tcs/Ishii18}, the call-by-value $\l$-calculus by Ronchi Della Rocca and Paolini \cite{DBLP:journals/iandc/PaoliniR04,RonchiPaolini}, the resource $\l$-calculus by Pagani and Tranquilli \cite{DBLP:conf/aplas/PaganiT09}, pattern calculi by Kesner, Lombardi and  \rios \cite{DBLP:journals/corr/abs-1102-3734}, the shuffling calculus by Guerrieri, Paolini and Ronchi Della Rocca \cite{GuerrieriPR15,Guerrieri15,DBLP:journals/lmcs/GuerrieriPR17}, and it has 
been formalized with proof assistants by McKinna and Pollack \cite{DBLP:journals/jar/McKinnaP99} and Crary \cite{Crary09standard}.

\paragraph{Takahashi revisited.} Despite its simplicity, Takahashi's proof \cite{DBLP:journals/iandc/Takahashi95} of factorization relies on substitutivity properties not satisfied by \complete reductions such as \lo and \ll. 
Our first contribution is a proof that is independent 
of the substitutivity properties of the factorizing reductions. 
It relies on a simpler fact, namely the substitutivity of an \emph{indexed} variant $\partobind n$ of parallel $\beta$ reduction $\partob$. 
The definition of $\partobind n$ simply decorates the definition of $\partob$ with a natural number $n$ that intuitively corresponds to the number of redexes reduced in parallel by a $\partob$ step.
 
We 
prove factorization theorems for all our four case studies following this simpler scheme. We also highlight  an interesting point: 
factorization for the two \complete reductions cannot be 
obtained directly following Takahashi's method\footnotemark
\footnotetext{\label{f:direct}It can be obtained indirectly, as a corollary of standardization, proved by Takahashi \cite{DBLP:journals/iandc/Takahashi95} using the concrete structure of terms. Thus the proof is not of an abstract nature.}. 

\paragraph{From factorization to essential normalization.} The second main  contribution of our  paper is the isolation of 
abstract properties  that together with  factorization imply normalization. First of all we abstract head reduction into a generic reduction $\toe$, called \emph{essential}, and non-head reduction $\tonoth$ into a \emph{non-essential} reduction $\toi$. The first additional property for normalization is \emph{persistence}: steps of the factoring reduction $\toe$ cannot be erased by the factored out $\toi$. The second one is a relaxed form of determinism for $\toe$. We show that in such \emph{essential} rewriting systems $\toe$ has a normalization theorem. The argument is abstract, that is, independent of the specific nature of terms. This is in contrast to how Takahashi \cite{DBLP:journals/iandc/Takahashi95} obtains normalization from factorization: her proof is based on an induction over the structure of terms, and cannot then be disentangled by the concrete nature of the rewriting system under study.

\paragraph{Normalizing reductions for $\beta$.} We  apply both our techniques to our case studies of  \complete reduction:  \lo and \ll, obtaining simple proofs that they are normalizing reductions for $\beta$. 
%
Let us point out that \lo is also---at present---the only known deterministic reduction to $\beta$ normal form whose number of steps is a reasonable cost model, as shown by 
Accattoli and Dal Lago \cite{DBLP:journals/corr/AccattoliL16}. Understanding its normalization is one of the motivations at the inception of this work.

\paragraph{Normalization with respect to different notions of results.}
As a further feature, our approach provides for free normalization theorems for reductions that are not \complete for the rewrite system in which they live. 
Typical examples are head and weak \cbv reductions, which do not compute $\beta$ and \cbv normal forms, respectively. 
These normalization theorems arise naturally in the theory of the $\l$-calculus. For instance, functional programming languages implement only weak notions of reduction, and head reduction (rather than \lo) is the key notion for the $\l$-definability of computable functions.

We obtain normalization theorems for head and weak \cbv reductions. 
Catching normalization for non-\complete reductions sets our work apart from the recent studies on normalization by Hirokawa, Middeldorp, and Moser \cite{DBLP:conf/rta/HirokawaMM15} and Van Oostrom and Toyama \cite{DBLP:conf/rta/OostromT16}, discussed below among related works.

\paragraph{Factorization, Normalization, Standardization.} In the literature of the $\l$-calculus, normalization for \LO reduction is often obtained as a corollary of the standardization theorem, which roughly states that every reduction sequence can be re-organized as to reduce redexes according to the left-to-right order (Terese \cite{Terese} following Klop \cite{phdklop} and Barendregt \cite{Barendregt84}, for instance). Standardization is a complex and technical result. 
Takahashi \cite{DBLP:journals/iandc/Takahashi95}, using Mitschke's argument \cite{Mitschke79} that iterates 
head factorization, obtains a simpler proof technique for \LO normalization---and for standardization as well. 
Our work refines that approach, abstracts from it and shows that factorization is a general technique for normalization.
%
%
%

\paragraph{Related work.} Factorization is studied in the abstract 
{in \cite{DBLP:conf/ctcs/Mellies97,DBLP:conf/rta/Accattoli12}}. 
\mellies axiomatic approach \cite{DBLP:conf/ctcs/Mellies97} builds on standardization, and encompasses a wide class of rewriting 
systems; 
in particular, like us, he can deal with non-\complete reductions. 
{Accattoli} \cite{DBLP:conf/rta/Accattoli12} relies crucially 
on terminating hypotheses, absent instead here. 

  Hirokawa, Middeldorp, and Moser \cite{DBLP:conf/rta/HirokawaMM15} and Van Oostrom and Toyama \cite{DBLP:conf/rta/OostromT16} study normalizing strategies  via a clean separation between abstract and term rewriting results. 
Our approach to normalization is similar to the one used in  \cite{DBLP:conf/rta/HirokawaMM15} to study \LO evaluation for first-order term rewriting systems. 
Our essential systems strictly generalize  their  conditions: uniform termination replaces determinism 
(two of the strategies we present here are not deterministic) and---crucially---persistence strictly generalizes the property in their Lemma 7. 
Conversely, they focus on hyper-normalization and on extending the method to systems in which left-normality is relaxed. We do not deal with these aspects. 
Van Oostrom and Toyama's study \cite{DBLP:conf/rta/OostromT16} of 
(hyper-)normalization is
based on an elegant and powerful method based on the random descent property and 
an ordered notion of commutative diagrams.
Their method and ours are incomparable: we do not rely on (and do not assume) the random descent property (for its definition and uses see van Oostrom \cite{DBLP:conf/rta/Oostrom07})---even if most strategies naturally have that property---and we do 
focus on factorization (which they explicitly 
avoid), 
since we see it as the crucial tool from  which normalization can be obtained.  

As already pointed out, a fundamental difference with respect to both works is that we consider a more general notion of normalization for reductions that are not \complete, that is not captured by either of those approaches.

In the literature, normalization is also proved from iterated head factorization (Takahashi \cite{DBLP:journals/iandc/Takahashi95} for \lo, and Terui \cite{DBLP:journals/aml/Terui07} or Accattoli \cite{DBLP:conf/rta/Accattoli12} for \ll on proof nets-like calculi, or Pagani and Tranquilli \cite{DBLP:journals/mscs/PaganiT17} for \ll on differential proof nets), or as a corollary of standardization (Terese \cite{Terese} following Klop \cite{phdklop} and Barendregt \cite{Barendregt84} for \lo), or using semantic principles such as intersection types (Krivine \cite{krivine1993lambda} for \lo and de Carvalho, Pagani and Tortora de Falco \cite{DBLP:journals/tcs/CarvalhoPF11} for \ll on proof nets). 
Last, Bonelli \emph{et al.} develop a sophisticated proof of normalization for a $\l$-calculus with powerful pattern matching in \cite{DBLP:journals/tcs/BonelliKLR17}.
Our technique differs from them all.

\paragraph{Proofs.} 
Most proofs are in the appendix. This paper is an extended version of~\cite{AccattoliFG19}.


\section{Factorization and Normalization, Abstractly}
\label{sect:abstract-rewriting}

In this section, we study factorization and normalization abstractly, that is,  independently of the specific structure of the objects to be rewritten. 

A \emph{rewriting system} (aka abstract reduction system, see Terese \cite[Ch.~1]{Terese})  $\rewsys$ is a pair $(\terms, \to)$ 
consisting of a  set $\terms$ and a binary 
relation $\to \, \subseteq \terms\times \terms$ called \emph{reduction}, whose pairs are written  $t \to s$ and called \emph{$\to$-steps}. 
A \emph{$\to$-\sequence} from $\tm$ is a sequence $\tm \to \tmtwo \to \dots$ of $\to$-steps; $\tm \to^k \tms$ denotes a sequence of $k$ $\to$-steps from $\tm$ to 
	$\tmtwo$. 
As usual, $\to^*$  (resp. $\to^=$) denotes the  transitive-reflexive (resp. reflexive)  closure of $\to$. Given two reductions $\toone$ and $\totwo$ we use $\toone\!\cdot\!\totwo$ for their composition, defined as $\tm \toone \!\cdot\! \totwo \tmtwo$ if $\tm \toone\tmthree\totwo \tmtwo$ for some $\tmthree$.

In this section we focus on a given sub-reduction $\toe$ of $\to$, called \emph{essential}, for which we study factorization and normalization with respect to $\to$. It comes with a second sub-reduction $\toi$, called \emph{inessential}, such that $\toe \cup \toi \,=\, \to$. Despite the notation, $\toe$ and $\toi$ are not required to be disjoint. In general, we write $(\terms,\{\Rew{a},\Rew{b}\})$ for the rewriting system 
$(\terms,\to)$ where $\to ~=~\Rew{a} \cup \Rew{b}$.

\subsection{Factorization.}

A rewriting system $(\terms,\{\toe,\toi\})$ 
satisfies \emph{$\toe$-factorization} (also called \emph{postponement of $\toi$ after $\toe$})
if $\tm \to^* \tmtwo$ implies that there exists $\tmthree$ such that $\tm\toe^* \tmthree \toi^* \tmtwo$.
Compactly, we write $\to^* \, \subseteq \, \toe^*\!\cdot\!\toi^*$. 
In diagrams, see \reffig{diagrams}.a.

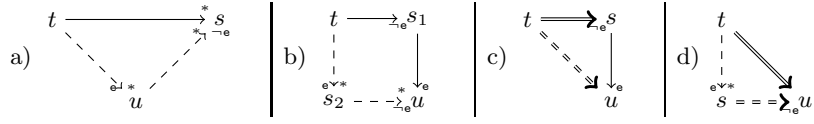
\begin{figure}[t]
\centering
\begin{tabular}{cc \colspace |  cc \colspace |  cc \colspace |  cc}
a)&
	\begin{tikzpicture}[ocenter]
	\node (t) {\small $\tm$};
	\node (dummy) [right of=t] {};
	\node (s1) [right of=dummy] {\small$\tmtwo$};
	\node (u) [below of=dummy] {\small$\tmthree$};
	\draw[->, labelEndAbove=*] (t) to (s1);
	\draw[->,dashed, labelEndLeft=$\extsym$, labelEndRight=*] (t) to (u);
	\draw[->, dashed, labelEndRight=$\intsym$, labelEndLeft=*] (u) to  (s1);
	\end{tikzpicture}
&
\ b)
&
	\begin{tikzpicture}[ocenter]
	\node (t) {\small $\tm$};
	\node (s1) [right of=t] {\small$\tmtwo_1$};
	\node (s2) [below of=t] {\small$\tmtwo_2$};
	\node (u) [right of=s2] {\small$\tmthree$};
	\draw[->, labelEndBelow=$\intsym$] (t) to (s1);
	\draw[->, dashed, labelEndLeft=$\extsym$, labelEndRight=*] (t) to (s2);
	\draw[->, dashed, labelEndAbove=*, labelEndBelow=$\intsym$] (s2) to  (u);
	\draw[->, labelEndRight=$\extsym$] (s1) to  (u);
	\end{tikzpicture}
&
\ c)
&
			\begin{tikzpicture}[ocenter]
			\node (t) {\small $\tm$};
			\node (s1) [right of=t] {\small$\tmtwo$};
			\node (u) [below of=s1] {\small$\tmthree$};
			\draw[->, double, labelEndBelow=$\intsym$] (t) to (s1);
			\draw[->, double, dashed, labelEndAbove] (t) to  (u);
			\draw[->, labelEndRight=$\extsym$] (s1) to  (u);
			\end{tikzpicture}  
&
\ d)
&
			\begin{tikzpicture}[ocenter]
			\node (t) {\small $\tm$};
			\node (s2) [below of=t] {\small$\tmtwo$};
			\node (u) [right of=s2] {\small$\tmthree$};
			\draw[->, double] (t) to (u);
			\draw[->, dashed, labelEndLeft=$\extsym$, labelEndRight=*] (t) to (s2);
			\draw[->, double, dashed, labelEndBelow=$\intsym$] (s2) to  (u);
			\end{tikzpicture}

	\end{tabular}
\caption{Diagrams: a) factorization, b) weak postponement, c) merge, d) split.}
\label{fig:diagrams}
\end{figure}

\paragraph{Proving factorization.}

Factorization is a non-trivial rewriting property, because it is \emph{global}, that is, quantified over all reduction sequences from a term. 
To be able to  prove factorization, we would like to reduce it  
to \emph{local} properties, \ie properties quantified only over one-step reductions from 
{a} term.
At first sight it may seem that a local diagram such as the one in \reffig{diagrams}.b would give factorization by a simple induction. Such a diagram however does not allow to infer factorization without further hypotheses---counterexamples can be found in 
Barendregt \cite{Barendregt84}.

The following abstract property is a special case for which 
a local condition implies factorization. 
It was first observed by Hindley \cite{HindleyPhD}. 

\begin{lemma}[Hindley, local postponement]
	\label{l:hindley}
	Let $(\terms,\{\toe,\toi\})$ be a rewriting system. If $\toi \!\cdot\! \toe \,\subseteq\, \toe^* \!\cdot\! \toi^=$ then  $\to^* \,\subseteq\, \toe^* \!\cdot\! \toi^*$. 	
\end{lemma}

\begin{proof}
	The assumption   $\toi \!\cdot\! \toe \,\subseteq\, \toe^* \!\cdot\! \toi^=$ implies (\#) $\toi \!\cdot\! \toe^* \,\subseteq\, \toe^* \!\cdot\! \toi^=$ (indeed, it is immediate to prove that $\toi \!\cdot\! \toe^k \,\subseteq\, \toe^* \!\cdot\! \toi^=$ by induction on $k$).
		
		We then prove that  $\to^k \,\subseteq\, \toe^* \!\cdot\! \toi^*$, by induction on  $k$.  The case $k=0$ is trivial. Assume $ \to \!\cdot\! \to^{k-1}$. By \ih, $ \to \!\cdot\! \toe^* \!\cdot\! \toi^*$. 
		If the first step is $\toe$, the claim is proved. Otherwise, 
		by $(\#)$, from  $(\toi \!\cdot\! \toe^*) \cdot\! \toi^*$ we obtain  
		$(\toe^* \!\cdot\! \toi^=) \cdot\! \toi^*$.
\end{proof}

Hindley's local condition 
is a strong hypothesis for factorization that in general does not hold in $\l$-calculi---not even in the simple case of head \strategy. 
However, the property can be applied in combination with another standard technique: switching to \emph{macro} steps that compress $\toe^* $ or $\toi^*$ into just one step, at the price of some light overhead. 
This idea is the essence of both Tait--Martin-L\"of's and Takahashi's techniques,  based on \emph{parallel steps}.
The role of parallel steps in 
Takahashi \cite{DBLP:journals/iandc/Takahashi95} 
is here captured abstractly  by the notion of  \macrostep\ system.

\begin{definition}[Macro-step  system] 
\label{def:macrostep}
A rewriting system $\rewsys=(\terms, \{\toe, \toi\})$ is a \emph{\macrostep\ system} if there are two 
reductions $\parto$ and $\partoint$ (called \emph{\macrostep s} and \emph{inessential \macrostep s}, respectively)  such that
	\begin{itemize}
		\item \emph{Macro}: 
		$\toi ~\subseteq~ \partoint~ \subseteq ~\toi^*$.
		
		\item \emph{Merge}: if $\tm \partoint \!\cdot\!  \toe \tmu $ then $\tm \parto \tmu$. That is, the diagram in \reffig{diagrams}.c holds.
		
		\item \emph{Split}:  if $\tm \parto \tmu$ then  $\tm \toe^* \!\cdot\! \partoint  \tmu$. That is,  the diagram in \reffig{diagrams}.d holds.
	\end{itemize}
\end{definition}
Note that $\parto$ just plays the role of a ``bridge'' between the hypothesis of the merge condition and the conclusion of the split condition---it shall play a crucial role in the concrete proofs in the next sections. 
In this paper, concrete instances of $\parto$ and $\partoint$ shall be parallel $\beta$ reduction and some of its variants. 

\begin{proposition}[Factorization]
	\label{prop:parallelsplit-gives-fact}
	Every \macrostep\ system $(\terms, \{\toe, \toi\})$ satisfies  $\toe$-factorization.
\end{proposition}
\begin{proof}
	By  Merge 
	and Split, 
	$\partoint \!\cdot\! \toe \,\subseteq\,  \parto \,\subseteq\, \toe^* \!\cdot\! \partoint \,\subseteq\, \toe^* \!\cdot\! \partoint^=$.  
	By Hindley's lemma (\reflemma{hindley}) applied to 
	$\toe$ and $\partoint$ (rather than $\toe$ and $\toi$), 
	we 
	obtain 
	$(\toe \cup \partoint)^* \,\subseteq\, \toe^* \!\cdot\!  \partoint^*$.
	Since $\toi \,\subseteq\, \partoint$, we have $ (\toe \cup \toi)^*  \subseteq\, (\toe \cup \partoint)^* \subseteq\,  \toe^* \!\cdot\! \partoint^* $. 
	As $\partoint \,\subseteq\, \toi^*$, we have  $\toe^* \!\cdot\! \partoint^* \, \subseteq \, \toe^* \!\cdot\! \toi^*$. 
	Therefore, $\to^* = (\toe \cup \toi)^* \subseteq\, \toe^* \!\cdot\! \toi^*$.
\end{proof}


\subsection{Normalization for \complete \strategies}
The interest of  factorization comes from the fact that the essential 
reduction $\toe$ on which factorization pivots has some good properties. 
Here we pinpoint the abstract  properties which make factorization a privileged method to prove normalization; we collect them into the definition of \emph{\external system} (\refdef{extsys}).

\paragraph{Normal forms and normalization.} Let us recall what normalization is about. 
In general, a term may or may not reduce to a normal form. 
And if it does, not all reduction sequences necessarily lead to normal form. 
A term is 
\emph{weakly} or \emph{strongly normalizing}, depending on if it  may or must reduce to normal form.
If a term $\tm$ is strongly normalizing, any choice of steps will eventually lead to a normal form. 
However, if $\tm$ is weakly normalizing, how do we compute a normal form? This is the problem tackled by \emph{normalization}:  by repeatedly performing \emph{only specific  steps},  a normal form will be computed, provided that $\tm $ can reduce to~any.

Recall the statement of the \LO normalization theorem: if $\tm \tob^* \tmu$ with $\tmu$ $\beta$-normal, then $\tm$ \LO-reduces to $\tmu$.
Observe  a subtlety: 
such a formulation relies on  the determinism of \LO reduction. 
We give  a more general formulation of normalizing reduction, valid also for non-deterministic reductions.\\

Formally, given a rewriting system $(\terms, \to)$, \mbox{a term $\tm \in \terms$ is:}
\begin{itemize}
	\item \emph{$\to$-normal} (or in \emph{$\to$-normal form})  if $\tm \not\to$, \ie there are no $\tms$ such that $\tm\to \tms$;
	\item \emph{weakly $\to$-normalizing} if \emph{there exists} a \sequence  $\tm \to^*\tmtwo$ with $\tmtwo$  $\to$-normal; 
	\item \emph{strongly $\to$-normalizing}  if there are no infinite $\to$-\sequence{s} from $\tm$, or 
	equivalently, if \emph{all} maximal  $\to$-\sequence{s} from $\tm$  are finite. 
\end{itemize}
We call \emph{$\strategy$ for} $\to$ any $\toe \,\subseteq\, \to$.
It is \emph{\complete}  if $ \toe$ and $\to$ have the same normal forms.\footnotemark
\footnotetext{In  rewriting theory, a \emph{\complete \strategy for $\to$} is called a \emph{reduction strategy for} $\to$. We prefer not to use the term strategy because it has different meaning  in the $\l$-calculus, where it is a  \emph{deterministic, not necessarily \complete}, reduction for $\to$.} 

\begin{definition}[Normalizing \strategy]
	A  \complete \strategy $\toe$ for $\to$  is  \emph{normalizing} (for $\to$) if, for every term $\tm$, $\tm$ is strongly $\toe$-normalizing whenever it is weakly $\to$-normalizing.
\end{definition}
%
Note that, since 
the normalizing \strategy $\toe$ is \complete, if  $\tm$ is strongly $\toe$-normalizing then \emph{every} maximal $\toe$-\sequence from $\tm$ ends in a $\to$-normal form.

\begin{definition}[\External system]
	\label{def:extsys} 
	A rewriting system   $(\terms,\{\toe,\toi\})$ is \emph{\external} if the following conditions hold:
	\begin{enumerate}
		\item \emph{Persistence:} if $\tm \toe \tmtwo$ and $\tm \toi \tmthree$, then $\tmthree \toe \tmfour$ for some $\tmfour$.
		\item  \emph{Uniform termination}: if $\tm$ is weakly $\toe$-normalizing, then it is  strongly \mbox{$\toe$-normalizing}.		
		\item  \emph{Terminal factorization}:  if $\tm\to^* \tmu$ and $\tmu$ is $\toe$-normal, then $ \tm\toe^* \!\cdot\! \toi^*\tmu$.
	\end{enumerate}
	It is moreover \emph{\complete} if $\toe$ is a \complete \strategy for $\to$.
\end{definition}
Comments on the definition:
\begin{itemize}
	\item \emph{Persistence}: it 
	means that essential steps are out of reach for inessential steps, that cannot erase them. 
	The only way of getting rid of essential steps is by reducing them, and so in that sense they are \emph{essential} to normalization.
	

	\item \emph{From determinism to uniform termination}: as we already said, in general $\toe$ is not deterministic. For normalization, then,
		 it is not enough that there is a \sequence $\tm \toe^*\tmu$ with $\tmu$ $\to$-normal (as in the statement of \LO-normalization). We need  to be sure that there are no 
		infinite $\toe$-\sequence{s} from $\tm$.
		This is exactly what is ensured by the uniform termination property. 
		Note that if $\toe$ is deterministic (or has the diamond or random descent properties) then it is  uniformly terminating.

	\item \emph{Terminal factorization}: there are two subtleties. First, we need only a weak form of factorization, namely 
	factorization is only required for $\to$-\sequence{s} ending in a $\toe$-normal form\footnote{The difference between factorization and its terminal case is relevant for normalization: 
		van Oostrom and Toyama \cite[footnote 8]{DBLP:conf/rta/OostromT16} give an example of normalizing \complete \strategy for a rewriting system in  which factorization fails 
		but terminal factorization holds.}.
	Second, the reader  may 
	expect terminal factorization to be required with respect to $\to$-normal rather than $\toe$-normal forms.
	The two notions coincide if $\toe$ is \complete, and for the time being we only discuss \complete essential systems. 
	We discuss the more general case in \refsect{non-complete}.
\end{itemize}

\begin{example} In the $\l$-calculus with $\beta$ reduction, head reduction $\toh$  and its associated $\tonoth$ reduction (defined in \refsect{head}) form an \external system. Similarly, leftmost-outermost $\tolo$ reduction and its associated $\tonotlo$ reduction (\refsect{leftmost}) form a \complete \external system. Two more examples are in \refsect{weak-cbv} and \refsect{stratified-cbn}.
\end{example}

\begin{theorem}[Essential \complete normalization]
	\label{thm:complete-normalization}
	Let 
	$(\terms,\{\toe,\toi\})$ be a \complete \external system. Then	$\toe$ is a  normalizing \strategy for $\to$.	
\end{theorem}
\begin{proof}
	Let $\tm$ be a weakly $\to$-normalizing term, \ie $\tm \to^* \tmu$ for some term $\tmu$ in $\to$-normal form (and so in $\toe$-normal form, since $\toe \,\subseteq\, \to$). 
	\begin{enumerate}
		
		\item  \emph{Terminal factorization} implies  $\tm \toe^*\tms \toi^* \tmu$ for some $\tmtwo$, since $\tmthree$ is $\toe$-normal.
		\item Let us show that $\tms$ is $\toe$-normal: if not, then $\tms\toe\tmfour$ for some $\tmfour$, and a straightforward induction on the length of $\tms \toi^*\tmu$ iterating \emph{persistence} gives that $\tmu \toe \tmfive$ for some $\tmfive$, against the hypothesis that $\tmu$ is $\toe$-normal. Absurd.

		\item By the previous point, $\tm$ is weakly $\toe$-normalizing. By \emph{uniform termination}, $\tm$ is strongly $\toe$-normalizing. 
		\qedhere
	\end{enumerate}
\end{proof}
%
%
%

\subsection{A more general notion of normalizing \strategy.}
\label{sect:non-complete}
Essential systems actually encompass also important notions of normalization for 
reductions that are \emph{not \complete}, 
such as \emph{head normalization}. 
These cases arise naturally in the $\l$-calculus literature, where partial notions of result such as head normal forms or values are of common use. 
Normalization for non-\complete reductions is instead not so common in the rewriting literature outside the $\l$-calculus. This is why, to guide the reader, we presented first the natural case of \complete reductions.

Let us first discuss head reduction:  $\toh$  is deterministic and not \complete with respect to $\tob$, as its normal forms may not be $\tob$-normal forms. The well-known property of interest is \emph{head normalization} (Cor.~11.4.8 in Barendregt's book \cite{Barendregt84}):
 \begin{center}	
			 If $\tm \tob^* \tmtwo$ and $\tmtwo$ is head normal\footnotemark
			 \footnotetext{``$\tm$ has a head normal form'' is the usual formulation for ``$\tm \tob^* \tmtwo$ for some $\tmtwo$ that is head normal''. We prefer the latter to avoid the ambiguity of the former about the reduction leading from $\tm$ to one of its head normal forms ($\tob^*$ or $\toh^*$?).} 
			 then $\toh$ terminates on $\tm$.
	\end{center}
\commCF{The above is Barendregt statement. Takahashi statement is:
 \begin{center}
 	 \begin{center}
 			if $\tm$ has a head normal form, then $\tm \toh^* \tmu$ for some $\tmu$ in head normal form. 
 	\end{center}
	if $\tm$ has a head normal form, iff the $\toh$-reduction from $\tm$ terminates.
\end{center}
}
 The statement has two subtleties. First, 
 $\tm$ may $\tob$-reduce to a term  in $\toh$-normal form in many different ways, 
 possibly without using $\toh$, so that 
 the hypotheses may not imply that $\toh$ terminates. 
 Second, 
 the conclusion is ``\emph{$\toh$ terminates on $\tm$}'' and not $\tm\toh^*\tmtwo$, because in general 
 the maximal $\toh$-sequence from $\tm$ may end in a term $\tmthree\neq\tmtwo$. 
 For instance, let $I=\la \vartwo \vartwo$: then 
	 $I(\var(II))\tob I(\var I)\tob \var I$ is a $\tob$-\sequence to head normal form, and yet 
	 the maximal $\toh$-sequence $I(\var(II))\toh \var(II)$ ends in a different term.

Now, let us abstract from head normalization, taking into account that in general the essential \strategy $\toe$---unlike 
head \strategy---may not be deterministic, and so we 
ask for strong $\toe$-normalization rather than for $\toe$-termination.

\begin{theorem}[Essential normalization]
	\label{thm:par-normalization}   
	Let 
	$(\terms,\{\toe,\toi\})$ be an  \external system. If 
	$\tm \to^* \tmthree$ and $\tmthree$ is $\toe$-normal, then $\tm$ is strongly $\toe$-normalizing. 	
\end{theorem}
\begin{proof} Exactly as for \refthm{complete-normalization}, \completeness is not used in that proof.
\end{proof}
 
In the next section we shall apply \refthm{par-normalization} to head reduction and obtain the head normalization theorem we started with. Another example of a normalization theorem for a non-\complete reduction is 
in \refsect{weak-cbv}.
Note that the \complete variant of the theorem (\refthm{complete-normalization}) is in fact an instance of the general one (\refthm{par-normalization}). 

\section{\texorpdfstring{The $\l$-Calculus}{The lambda-Calculus}}
\label{sect:lambda}

This short section recalls basic definitions and properties of the $\l$-calculus and introduces the indexed variant of parallel $\beta$.

The set $\lterms$ of \emph{terms} of the $\l$-calculus is  given by the following grammar:
\[ \begin{array}{c@{\hspace{.5cm}}rll}
	    	    \textsc{Terms} & \tm,\tmtwo,\tmthree,\tmfour &\grameq& \var \mid  \la\var\tm \mid \tm\tmtwo
    \end{array}\]
We use the usual notions of free and bound variables, $\tm \isub\var\tmtwo$ for the meta-level capture-avoiding substitution of $\tmtwo$ for the free occurrences of $\var$ in $\tm$, and $\sizep\tm\var$ for the number of free occurrences of $\var$ in $\tm$.
The definition of $\beta$ reduction $\tob$ is:
\[
\begin{array}{c@{\ \sep}c@{\ \sep}c@{\ \sep}c}
 \multicolumn{4}{c}{\textsc{$\beta$ reduction} }
 \\[5pt]
	\AxiomC{}
	\UnaryInfC{$(\la\var\tm) \tmtwo \tob \tm \isub\var{\tmtwo}$}
	\DisplayProof
	&
	\AxiomC{$\tm \tob \tm'$}	
	\UnaryInfC{$\tm  \tmtwo \tob \tm'  \tmtwo$}
	\DisplayProof 
 &
	
	\AxiomC{$\tm \tob \tm'$}
	\UnaryInfC{$\la\var\tm\ \tob \la\var\tm' $}
	\DisplayProof 	
	&
	\AxiomC{$\tm \tob \tm'$}	
	\UnaryInfC{$\tmtwo  \tm \tob \tmtwo  \tm'$}
	\DisplayProof 
\end{array}\]
Let us recall two basic substitutivity properties of $\beta$ reduction.
\begin{enumerate}  
    \item \label{p:tob-subs-one} 
    \emph{Left substitutivity of $\tob$}:
    if $\tm \tob \tm'$ then $\tm\isub\var\tmtwo \tob \tm'\isub\var\tmtwo$.
    \item \label{p:tob-subs-two}
    \emph{Right substitutivity of $\tob$}:
    if $\tmtwo \tob \tmtwo'$ then $\tm\isub\var\tmtwo \tob^* \tm\isub\var{\tmtwo'}$. It is possible to spell out the number of $\tob$-steps, which is exactly the number of free occurrences of $\var$ in $\tm$, that is, $\tm\isub\var\tmtwo \tob^{\sizep\tm\var} \tm\isub\var{\tmtwo'}$.
\end{enumerate}
\paragraph{Parallel $\beta$ reduction.} Parallel $\beta$-reduction $\partob$ is defined by:
\begin{center}
		\def\ScoreOverhang{1pt}
$\begin{array}{c@{\ \sep}c@{\ \sep}c@{\ \sep}c}
 \multicolumn{4}{c}{\textsc{Parallel $\beta$ reduction}}
 \\[5pt]
\AxiomC{}
\UnaryInfC{$\var \partob \var$}
\DisplayProof
&
\AxiomC{$\tm \partob \tm'$}
\UnaryInfC{$\la \var\tm \partob  \la \var \tm'$}
\DisplayProof
&
\AxiomC{$\tm \partob \tm'$}
\AxiomC{$\tmtwo \partob \tmtwo'$}
\BinaryInfC{$\tm \tmtwo \partob  \tm' \tmtwo'$}
\DisplayProof
&
	\AxiomC{$\tm \partob  \tm'$}
	\AxiomC{$\tmtwo \partob \tmtwo'$}
\BinaryInfC{$(\la\var \tm)\tmtwo \partob \tm'\isub \var {\tmtwo'}$}
\DisplayProof
\end{array}$
	\def\ScoreOverhang{4pt}
\end{center}
Tait--Martin-L\"of's proof of the confluence of $\tob$ relies on the 
diamond property of $\partob$\footnotemark,
\footnotetext{Namely, if $\tmtwo_1 \lpartob \tm \partob \tmtwo_2$ 
	then there exists $\tmthree$ such that $\tmtwo_1 \partob \tmthree \lpartob \tmtwo_2$.} in turn based on the 
following property (see Takahashi~\cite[p.~1]{DBLP:journals/iandc/Takahashi95})
\begin{center}
	 \emph{Substitutivity of $\partob$}:
	if  $\tm \partob \tm'$ and  $\tmtwo \partob \tmtwo'$ then $\tm\isub\var\tmtwo \partob \tm'\isub\var{\tmtwo'}$. 
\end{center}
While the diamond property of $\partob$ does not play a role for factorization, one of the contributions of this work is a new proof technique for factorization relying on the substitutivity property of an indexed refinement of $\partob$.

\paragraph{Indexed parallel $\beta$ reduction.}
The new \emph{indexed} version $\partobind n$ of parallel $\beta$ reduction $\partob$ is equipped 
  with  a natural number $n$ which is, roughly, the number of redexes reduced in parallel by a $\partob$; more precisely, $n$ is the length of a particular way of sequentializing the redexes reduced by $\partob$.
  The definition of $\partobind{n}$ is as follows (note that erasing the index one obtains exactly $\partob$, so that $\partob \,=\, \bigcup_{n \in \nat} \partobind{n}$):
\begin{center}
		\def\ScoreOverhang{1pt}
$\begin{array}{c@{\ \sep}c@{\ \sep}c@{\ \sep}c}
 \multicolumn{4}{c}{\textsc{Indexed parallel $\beta$ reduction}}
 \\[5pt]
\AxiomC{}
\UnaryInfC{$\var \partobind 0 \var$}
\DisplayProof
&
\AxiomC{$\tm \partobind n \tm'$}
\UnaryInfC{$\la \var\tm \partobind n  \la \var \tm'$}
\DisplayProof
&
\AxiomC{$\tm \partobind n \tm'$}
\AxiomC{$\tmtwo \partobind m \tmtwo'$}
\BinaryInfC{$\tm \tmtwo \partobindlong {n + m} \tm' \tmtwo'$}
\DisplayProof
&
	\AxiomC{$\tm \partobind n \tm'$}
	\AxiomC{$\tmtwo \partobind m \tmtwo'$}
\BinaryInfC{$(\la\var \tm)\tmtwo \partobindlong {n + \sizep{\tm'}\var \cdot m +1} \tm'\isub \var {\tmtwo'}$}
\DisplayProof
\end{array}$
	\def\ScoreOverhang{4pt}
\end{center}

The intuition behind the last clause 
is: $(\la\var \tm)\tmtwo$ reduces to $\tm'\isub \var {\tmtwo'}$ by 
\begin{enumerate}
\item first reducing $(\la\var \tm)\tmtwo$ to $\tm\isub \var {\tmtwo}$ ($1$ step);
\item then reducing in $\tm\isub \var {\tmtwo}$ the $n$ steps corresponding to the \sequence $\tm \partobind n \tm'$, obtaining $\tm'\isub \var {\tmtwo}$;
\item then reducing $\tmtwo$ to $\tmtwo'$ for every occurrence of $\var$ in $\tm'$ replaced by $\tmtwo$, that is, $m$ steps $\sizep{\tm'}\var$ times, obtaining $\tm'\isub \var {\tmtwo'}$.
\end{enumerate}
Points 2 and 3 hold because of the substitutivity properties of $\beta$ reduction.

It is easily seen that $\partobind 0$ is 
the identity relation on terms.
Moreover, $\tob \ =\ \partobind 1$, and $\partobind n \ \subseteq \ \tob^n$, as expected.  
The substitutivity of $\partobind n$ is proved by simply indexing the proof of substitutivity of $\partob$. 

\begin{lemma}[Substitutivity of $\partobind n$]
\label{l:partobind-subs} 
If $\tm \partobind n \tm'$ and $\tmtwo \partobind m \tmtwo'$, then $\tm \isub\var\tmtwo \partobind{k} \tm' \isub\var{\tmtwo'}$ where $k = n + \sizep{\tm'}\var\cdot m$.
\end{lemma}

\begin{proof}
	By induction on the derivation of $\tm \partobind n \tm'$. 
	Consider its last rule.
	Cases:
	\begin{itemize}
		\item \emph{Variable}: two sub-cases
		\begin{itemize}
			\item $\tm = \var$: then $\tm = \var \partobind 0 \var = \tm'$ then $\tm \isub\var\tmtwo = \var\isub\var\tmtwo = \tmtwo \partobind m \tmtwo' = \var \isub\var{\tmtwo'} = \tm' \isub\var{\tmtwo'}$ that satisfies the statement because $n + \sizep{\tm'}\var\cdot m = 0+1\cdot m = m$. 
			\item $\tm = \vartwo$: then $\tm = \vartwo \partobind 0 \vartwo = \tm'$ and $\tm \isub\var\tmtwo = \vartwo\isub\var\tmtwo = \vartwo \partobind 0 \vartwo = \vartwo \isub\var{\tmtwo'} = \tm' \isub\var{\tmtwo'}$ that satisfies the statement because $n + \sizep{\tm'}\var\cdot m = 
			0 + 0 \cdot m = 0$.
		\end{itemize}
		
		\item \emph{Abstraction}, \ie $\tm = \la{\vartwo}{\tmthree} \partobind{n} \la{\vartwo}{\tmthree'} = \tm'$ because $\tmthree \partobind{n} \tmthree'$; 
		we can suppose without loss of generality that $\vartwo \neq \var$ and $\vartwo$ is not free in $\tmtwo$ (and hence in $\tmtwo'$), so $\sizep{\tmthree'}{\var} = \sizep{\tm'}{\var}$ and $\tm\isub{\var}{\tmtwo} = \la{\vartwo}(\tmthree\isub{\var}{\tmtwo})$ and $\tm'\isub{\var}{\tmtwo'} = \la{\vartwo}(\tmthree'\isub{\var}{\tmtwo'})$.
		By \ih, $\tmthree\isub{\var}{\tmtwo} \partobindlong{n +  \sizep{\tmthree'}{\var} \cdot m} \tmthree'\isub{\var}{\tmtwo'}$, thus
		\begin{align*}
		\AxiomC{$\tmthree\isub{\var}{\tmtwo} \partobindlong{n +  \sizep{\tmthree'}{\var} \cdot m} \tmthree'\isub{\var}{\tmtwo'}$}
		\UnaryInfC{$\tm\isub{\var}{\tmtwo} = \la {\vartwo}\tmthree\isub{\var}{\tmtwo} \partobindlong{n +  \sizep{\tm'}{\var} \cdot m} \la{\vartwo}\tmthree'\isub{\var}{\tmtwo'} = \tm'\isub{\var}{\tmtwo'}$}
		\DisplayProof\,.
		\end{align*}

		\item \emph{Application}, \ie	
		$\tm = \tmthree \tmfour \partobind{n} \tmthree' \tmfour'= \tm'$ with $\tmthree \partobind {n_\tmthree} \tmthree'$, $\tmfour \partobind {n_\tmfour} \tmfour'$ and $n = n_\tmthree + n_\tmfour$.
		By \ih, $\tmthree\isub\var\tmtwo \partobindlong {n_\tmthree +\sizep{\tmthree'}\var \cdot m} \tmthree'\isub\var{\tmtwo'}$ and $\tmfour\isub\var\tmtwo \partobindlong {n_\tmfour +\sizep{\tmfour'}\var \cdot m} \tmfour'\isub\var{\tmtwo'}$. Then 
		\[
			\def\ScoreOverhang{1pt}
			\def\defaultHypSeparation{\hskip .5in}
			\AxiomC{$\tmthree\isub\var\tmtwo \partobindlong {n_\tmthree +\sizep{\tmthree'}\var \cdot m} \tmthree'\isub\var{\tmtwo'}$}
		\AxiomC{$\tmfour\isub\var\tmtwo \partobindlong {n_\tmfour +\sizep{\tmfour'}\var \cdot m} \tmfour'\isub\var{\tmtwo'}$}
		\BinaryInfC{$\tm\isub\var\tmtwo = \tmthree\isub\var\tmtwo \tmfour\isub\var\tmtwo \partobind{k} \tmthree'\isub\var{\tmtwo'} \tmfour'\isub\var{\tmtwo'} = \tm'\isub\var{\tmtwo'}$}
		\DisplayProof
						\def\defaultHypSeparation{\hskip .2in}
			\def\ScoreOverhang{4pt}
			\]
	where $k = n_\tmthree + \sizep{\tmthree'}\var \!\cdot m + n_\tmfour + \sizep{\tmfour'}\var \!\cdot m = n + (\sizep{\tmthree'}\var + \sizep{\tmfour'}\var)\cdot m = n + \sizep{\tm'}\var \!\cdot m$.
		
		\item \emph{$\beta$-step}, \ie\
		$\tm = (\la\vartwo \tmthree)\tmfour \partobind{n} \tmthree'\isub \vartwo {\tmfour'} = \tm'$
		with $\tmthree \partobind {n_\tmthree} \tmthree'$, $\tmfour \partobind {n_\tmfour} \tmfour'$ and $n = n_\tmthree + \sizep{\tmthree'}\vartwo \cdot n_\tmfour +1$. 
		We can assume without loss of generality that $\vartwo\neq \var$ and $\vartwo$ is not free in $\tmtwo$ (and so in $\tmtwo'$),
		hence $\sizep{\tm'}{\var} = \sizep{\tmthree'\isub\vartwo {\tmfour'}}\var = \sizep{\tmthree'}{\var} + \sizep{\tmthree'}{\vartwo}\cdot\sizep{\tmfour'}{\var}$ and $\sizep{\tmthree'\isub{\var}{\tmtwo'}}{\vartwo} = \sizep{\tmthree'}{\vartwo}$ and $\tm\isub\var\tmtwo = (\la\vartwo\tmthree\isub\var\tmtwo) (\tmfour\isub\var\tmtwo)$ and $\tm'\isub\var{\tmtwo'} = \tmthree'\isub\var{\tmtwo'} \isub\vartwo{\tmfour'\isub\var{\tmtwo'}}$.
		
		By \ih, $\tmthree\isub\var\tmtwo \partobindlong {n_\tmthree +\sizep{\tmthree'}\var \cdot m} \tmthree'\isub\var{\tmtwo'}$ and $\tmfour\isub\var\tmtwo \partobindlong {n_\tmfour +\sizep{\tmfour'}\var \cdot m} \tmfour'\isub\var{\tmtwo'}$. 
		Then 
		\[
					\def\defaultHypSeparation{\hskip .5in}
		\AxiomC{$\tmthree\isub\var\tmtwo \partobindlong {n_\tmthree +\sizep{\tmthree'}\var \cdot m} \tmthree'\isub\var{\tmtwo'}$}
		\AxiomC{$\tmfour\isub\var\tmtwo \partobindlong {n_\tmfour +\sizep{\tmfour'}\var \cdot m} \tmfour'\isub\var{\tmtwo'}$}
		\BinaryInfC{$\tm\isub\var\tmtwo = (\la\vartwo\tmthree\isub\var\tmtwo) (\tmfour\isub\var\tmtwo) 
		\partobind{k}
		 \tmthree'\isub\var{\tmtwo'} \isub\vartwo{\tmfour'\isub\var{\tmtwo'}}= \tm'\isub\var{\tmtwo'}$}
		\DisplayProof
					\def\defaultHypSeparation{\hskip .2in}
		\]
		
		where
		$k = n_\tmthree +\sizep{\tmthree'}\var \cdot m +\sizep{\tmthree'}\vartwo\cdot(n_\tmfour +\sizep{\tmfour'}\var \cdot m)+1 
		=
		n_\tmthree + \sizep{\tmthree'}\vartwo\cdot n_\tmfour + 1 + \sizep{\tmthree'}\var \cdot m + \sizep{\tmthree'}\vartwo\cdot \sizep{\tmfour'}\var \cdot m = n + \sizep{\tm'}\var\cdot m.$
		\qedhere
	\end{itemize}
\end{proof}

\section{Head Reduction, Essentially}
\label{sect:head}
We here revisit Takahashi's study \cite{DBLP:journals/iandc/Takahashi95} of head reduction. 
We apply the abstract schema for essential reductions developed in \refsect{abstract-rewriting}, which is the same schema used by Takahashi, but we provide a simpler proof technique for one of the required properties (split). 
First of all, head reduction $\toh$ 
(our essential reduction here)
and its associated inessential reduction $\tonoth$ are defined by:
\[\begin{array}{c}
\begin{array}{c\colspace c\colspace cccc}
 \multicolumn{3}{c}{\textsc{Head reduction}}
 \\[5pt]
	\AxiomC{}
	\UnaryInfC{$(\la\var \tm) \tmtwo \toh \tm\isub\var\tmtwo$}
	\DisplayProof
	&
	\AxiomC{$\tm \toh \tmtwo$}
	\AxiomC{$\tm\neq\la\var\tm'$}
	\BinaryInfC{$\tm \tmthree \toh \tmtwo \tmthree$}
	\DisplayProof 
	&	
	\AxiomC{$\tm \toh \tmtwo$}
	\UnaryInfC{$\la\var\tm \toh \la\var\tmtwo$}
	\DisplayProof 
\end{array}
\\[18pt]
\begin{array}{c\colspace c\colspace c\colspace c}
\multicolumn{4}{c}{\textsc{$\neg$Head reduction}}
\\[5pt]
\AxiomC{$\tm \tob \tm'$}		
		\UnaryInfC{$(\la\var \tm) \tmtwo \tonoth (\la\var \tm') \tmtwo$}
		\DisplayProof
		&
	\AxiomC{$\tm \tob \tm'$}
	\UnaryInfC{$\tmtwo \tm  \tonoth  \tmtwo \tm'$}
	\DisplayProof 
	&	
	\AxiomC{$\tm \tonoth \tm'$}
	\UnaryInfC{$\la\var\tm \tonoth \la\var\tm'$}
	\DisplayProof 
	&
		\AxiomC{$\tm \tonoth \tm'$}
	\UnaryInfC{$\tm \tmtwo \tonoth \tm' \tmtwo $}
	\DisplayProof \,.
\end{array}
\end{array}\]
Note that $\tob \, = \, \toh \cup \tonoth$ but $\toh$ and $\tonoth$ are not disjoint: $I (II) \toh II$ and $I(II) \tonoth II$ with $I = \la{\varthree}{\varthree}$.
Indeed,  $I(II)$ contains two distinct redexes, one is $I(II)$ and is fired by $\toh$, the other one is $II$ and is fired by $\tonoth$; 
coincidentally, the two reductions lead to the same term.

 As for Takahashi, a parallel $\neg$head step 
$\tm \partonoth \tmtwo$ is a parallel step $\tm \partob \tmtwo$ such that $\tm \tonoth^* \tmtwo$. We give explicitly the inference rules for $\partonoth$:
\begin{center}
\textsc{Parallel $\neg$head reduction}
\def\ScoreOverhang{1pt}
\def\defaultHypSeparation{\hskip .1in}
\begin{align*}
		\AxiomC{}
		\UnaryInfC{$\var \partonoth \var$}
		\DisplayProof
		\qquad
		\AxiomC{$\tm \partob \tm'$}
		\AxiomC{$\tmtwo \partob \tmtwo'$}
		\BinaryInfC{$(\la\var \tm) \tmtwo \partonoth (\la\var \tm') \tmtwo'$}
		\DisplayProof
		\qquad
		\AxiomC{$\tm \partonoth \tm'$}	
		\UnaryInfC{$\la\var\tm \partonoth \la\var\tm'$}
		\DisplayProof 
		\qquad
		\AxiomC{$\tm \partonoth \tm'$}
		\AxiomC{$\tmtwo \partob \tmtwo'$}
		\BinaryInfC{$\tm \tmtwo \partonoth \tm' \tmtwo'$}
		\DisplayProof 
	\end{align*}
	\def\ScoreOverhang{4pt}
	\def\defaultHypSeparation{\hskip .2in}
\end{center}
Easy inductions show that $\tonoth \,\subseteq\, \partonoth \,\subseteq\, \tonoth^*$.  It is immediate that $\toh$-normal terms are head normal forms 
in the sense of Barendregt \cite[Def. 2.2.11]{Barendregt84}. We do not describe the shape of head normal forms. 
Our proofs never use it, 
unlike Takahashi's ones. This fact stresses the abstract nature of our proof method.

\paragraph{Head factorization.} We show that $\toh$ induces a \macrostep\ system, with respect to $\tonoth$, $\partob$, and $\partonoth$, to obtain $\toh$-factorization by \refprop{parallelsplit-gives-fact}. 

Therefore, we need to prove merge and split. 
Merge is  easily verified by  induction on $\tm \partonoth\tmtwo$. 
The interesting part is the proof of the split property, that in the concrete case of head reduction becomes: if $\tm \partob \tmtwo$ then $\tm \toh^* \!\cdot\! \partonoth \tmtwo$. 
This is obtained as a consequence of the following easy \emph{indexed split} property based on the indexed variant of parallel $\beta$. The original argument of Takahashi \cite{DBLP:journals/iandc/Takahashi95} is more involved, we discuss it after the new proof.

\newcounter{prop:macro-head}
\addtocounter{prop:macro-head}{\value{proposition}}
\begin{proposition}[Head macro-step system]
\label{prop:macro-head}
\hfill
\begin{enumerate}
	\item\label{p:macro-head-merge} \emph{Merge}: if $\tm \partonoth \!\cdot\! \toh \tmthree$ then $\tm \partob \tmthree$.
 	
 	\item\label{p:macro-head-indexed-split} \emph{Indexed split}: if $\tm \partobind n \tmtwo$ then $\tm \partonoth\tmtwo$, or $n>0$ and $\tm \toh \!\cdot\! \partobind {n-1} \tmtwo$.

 	\item\label{p:macro-head-split} \emph{Split}: if $\tm \partob \tmtwo$ then $\tm \toh^* \!\cdot\! \partonoth\tmtwo$.
\end{enumerate}
That is, $(\Lambda, \{ \toh, \tonoth \})$ is a \macrostep\ system with respect to $\partob$ and $\partonoth$.
\end{proposition}

\begin{proof}
\begin{enumerate}
 	\item Easy induction on $\tm \partonoth\tmtwo$. 
 	Details are in \SLV{\cite{long}.}{the Appendix, p.~\pageref{propappendix:macro-head}. }

\item By induction on $\tm \partobind n \tmtwo$. 
We freely use the fact that if $\tm \partobind n \tmtwo$ then $\tm \partob \tmtwo$. Cases:
\begin{itemize}
	\item \emph{Variable}: $\tm = \var \partobind 0 \var = \tmtwo$. Then $\tm = \var \partonoth\var = \tmtwo$.
	
	\item \emph{Abstraction}: $\tm = \la \var\tm' \partobind n  \la \var \tmtwo' = \tmtwo$ 
	with $\tm' \partobind n \tmtwo'$. It follows from the~\ih
	
	\item \emph{Application}:
 $\tm = \tmfour \tmfive \partobind{n} \tmfour' \tmfive' = \tmtwo$ with $\tmfour \partobind {n_1} \tmfour'$, $\tmfive \partobind {n_2} \tmfive'$ and $n = n_1 + n_2$.
There are only two subcases:
		\begin{itemize}
			\item either $\tmfour \tmfive \partonoth \tmfour' \tmfive'$, and then the claim holds;
			\item or $\tmfour \tmfive \not\partonoth \tmfour' \tmfive'$, and then neither $\tmfour \partonoth \tmfour'$ nor $\tmfour$ is an abstraction (otherwise $\tmfour \tmfive \partonoth \tmfour' \tmfive'$). 
			By \ih applied to $\tmfour \partobind{n_1} \tmfour'$,
			$n_1>0$ and there exists $\tmfour''$ such that $\tmfour \toh \tmfour'' \partobind {n_1-1} \tmfour'$. 
			Thus, $\tm = \tmfour \tmfive \toh \tmfour'' \tmfive$ and 
			\[\AxiomC{$\tmfour'' \partobindlong {n_1-1} \tmfour'$}			
			\AxiomC{$\tmfive \partobind {n_2} \tmfive'$}	
			\BinaryInfC{$ \tmfour'' \tmfive \partobindlong {n_1-1 + n_2} \tmfour' \tmfive' = \tmtwo$}
			\DisplayProof
			\ .
			\]
		\end{itemize}
	
	\item \emph{$\beta$-step}:
	$\tm = (\la\var \tmthree)\tmfour \partobind{n} \tmthree'\isub \var {\tmfour'} = \tmtwo$ with $\tmthree \partobind {n_1} \tmthree'$, $\tmfour \partobind {n_2} \tmfour''$ and $n = n_1 + \sizep{\tmthree'}\var \cdot n_2 +1 > 0$.
	We have $\tm = (\la\var \tmthree)\tmfour \toh \tmthree\isub\var\tmfour$ and by substitutivity of $\partobind n$ (\reflemma{partobind-subs}) $\tmthree\isub\var\tmfour \partobindlong{n_1 + \sizep{\tmthree'}\var \cdot n_2} \tmthree' \isub\var{\tmfour'} = \tmtwo$.

\end{itemize}
 	\item If $\tm \partob \tmtwo$ then $\tm \partobind n \tmtwo$ for some $n$. We prove the statement by induction $n$. 
 	By \emph{indexed split} (\refpoint{macro-head-indexed-split}), there are only two cases:
	\begin{itemize}
	  \item \emph{$\tm \partonoth\tmtwo$}. This is an instance of the statement (since $\toh^*$ is reflexive).
	  
	  \item $n>0$ and there exists $\tmfour$ such that $\tm \toh \tmfour \partobind {n-1} \tmtwo$. 
	  By \ih applied to $\tmfour \partobind{n-1} \tmtwo$, there is $\tmthree$ such that $\tmfour \toh^* \tmthree \partonoth \tmtwo$, and so $\tm \toh^* \tmthree \partonoth\tmtwo$.
	  \qedhere
	\end{itemize}
\end{enumerate}
\end{proof}

\begin{theorem}[Head factorization]
\label{thm:head-fact}
 If $\tm \tob^* \tmu$ then $\tm\toh^* \!\cdot\! \tonoth^* \tmu$.
\end{theorem}

\paragraph{Head normalization.} We  
{show} that $(\Lambda, \{ \toh, \tonoth \})$ is an \external system (\refdef{extsys}); 
{thus} the essential normalization theorem (\refthm{par-normalization}) provides normalization. 
We already proved \emph{factorization} (\refthm{head-fact}, hence terminal factorization). 
We {verify} 
persistence and determinism (which implies uniform termination) of $\toh$. 

\newcounter{prop:essential-head}
\addtocounter{prop:essential-head}{\value{proposition}}
\begin{proposition}[Head essential system]
\label{prop:essential-head}
\NoteProof{propappendix:essential-head}
\begin{enumerate}
 \item\label{p:essential-head-persistence} \emph{Persistence}: if $\tm \toh \tmtwo$ and $\tm \tonoth \tmthree$ then $\tmthree \toh \tmfour$ for some $\tmfour$.
 \item\label{p:essential-head-determinism} \emph{Determinism}: if $\tm \toh \tmtwo_1$ and $\tm\toh \tmtwo_2$ then $\tmtwo_1 = \tmtwo_2$.
\end{enumerate}
Then, $(\Lambda, \set{ \toh, \tonoth})$ is an essential system.
\end{proposition}

\begin{theorem}[Head normalization]
		If $\tm \tob^* \tmtwo$ and $\tmtwo$ is a $\toh$-normal form, then $\toh$ terminates on $\tm$. 
\end{theorem}

\subsection{Comparison with Takahashi's proof of the split property.}\label{sect:Takahashi} Our technique differs from Takahashi's in that it is built on simpler properties: 
it exploits directly the substitutivity of $\partob$, which is instead not used by Takahashi. 
Takahashi's original argument \cite{DBLP:journals/iandc/Takahashi95} for the split property (\emph{if $\tm \partob \tmtwo$ then $\tm\toh^* \cdot \partonoth$}, what she calls the \emph{main lemma}) is by induction on the (concrete) definition of $\partob$ and relies on two substitutivity properties of $\toh$ and $\partonoth$. Looking at them as the reductions $\toe$ and $\toi$ of an essential system, these properties are:
  \begin{itemize}
    \item \emph{Left substitutivity of $\toe$}: if $\tmthree \toe \tmsix$ then $\tmthree \isub\var\tmfour \toe \tmsix \isub\var\tmfour$;
    \item \emph{Left substitutivity of $\partoint$}: if $\tmthree \partoint \tmsix$ then $\tmthree \isub\var\tmfour \partoint \tmsix \isub\var\tmfour$.
  \end{itemize}
From them, left substitutivity of the composed reduction $\toe^* \!\cdot\! \partoint$ easily follows. That is, Takahashi's proof of the split property  is by induction on $\tm \parto \tmtwo$ using left substitutivity of $\toe^*\cdot\partoint$ for the inductive case.

We 
exploit the substitutivity of $\partoind n$ instead of left substitutivity of $\toe$ and $\partoint$. 
It holds for a larger number of essential systems  because $\partoind{n}$ is simply a decoration of $\parto$, which is substitutive \emph{by design}. 
There are important systems where Takahashi's hypotheses do not hold. 
One such case is \lo reduction (\refsect{leftmost})---\emph{the} normalizing
	reduction of the $\l$-calculus---we discuss the failure of left substitutivity for \lo on p.~\pageref{thm:LO-fact-norm};
	another notable case is \ll reduction  (\refsect{stratified-cbn}); both  are \complete reductions for~$\beta$.

Let us 
point out where the idea behind our approach stems from. 
In a sense, Takahashi's proof works by chance: the 
split hypothesis is about a \emph{parallel} step $\partob$ but then the key fact used in the proof, left substitutivity of $\toh^* \!\cdot\! \partonoth$, does no longer stay in the borders of the parallel step, since
the prefix $\toh^*$ is an arbitrary long \sequence that may reduce \emph{created} steps. Our proof scheme instead only focuses on the (expected)
 substitutivity of $\partoind n$, independently of creations.

\section{\texorpdfstring{The Call-by-Value $\l$-Calculus}{The Call-by-Value lambda-Calculus}}
\label{sect:cbv-calculus}

In this short section, we introduce Plotkin's  call-by-value 
$\l$-calculus \cite{DBLP:journals/tcs/Plotkin75},
where
 $\beta$ reduction fires only when the argument is a value. In the next section we define \emph{weak} reduction and prove factorization and normalization theorems using the essential technique, exactly as done in the previous section for head reduction.
 
The set $\Lambda$ of terms is  the same as in \refsect{lambda}. Values,  call-by-value (\cbv) $\beta$-reduction $\tobv$, and 
\cbv indexed parallel reduction $\partobvind{n}$ are defined as follows:
 \begin{center}
 $\begin{array}{c@{\hspace{.5cm}}rll}
\textsc{Values} & \val &\grameq& \var \mid  \la\var\tm 
\end{array}$
\\[6pt]
$\begin{array}{c@{\ \sep}c@{\ \sep}c@{\ \sep}c}
 \multicolumn{4}{c}{\textsc{$\betav$ reduction}} 
 \\[5pt]
	\AxiomC{$\val$ value}
 	\UnaryInfC{$(\la\var\tm) \val \tobv \tm \isub\var{\val}$}
 	\DisplayProof
 	&
 		\AxiomC{$\tm \tobv \tm'$}
 	\UnaryInfC{$\la\var\tm\ \tobv \la\var\tm' $}
 	\DisplayProof 	
 	&
 	\AxiomC{$\tm \tobv \tm'$}	
 \UnaryInfC{$\tm  \tmtwo \tobv \tm'  \tmtwo$}
 \DisplayProof 
 	&
 	\AxiomC{$\tm \tobv \tm'$}	
 	\UnaryInfC{$\tmtwo  \tm \tobv \tmtwo  \tm'$}
 	\DisplayProof 
\end{array}$
\end{center}
\begin{center}
	
		\textsc{Indexed parallel $\betav$ reduction}\\[5pt]
	\small
	\def\ScoreOverhang{1pt}
	$\begin{array}{c@{\ \sep}c@{\ \sep}c@{\ \sep}c}
		\AxiomC{}
		\UnaryInfC{$\var \partobvind 0 \var$}
		\DisplayProof
		&
		\AxiomC{$\tm \partobvind n \tm'$}
		\UnaryInfC{$\la \var\tm \partobvind n  \la \var \tm'$}
		\DisplayProof
        &
		\AxiomC{$\tm \partobvind n \tm'$}
		\AxiomC{$\tmtwo \partobvind m \tmtwo'$}
		\BinaryInfC{$\tm \tmtwo \partobvindlong {n + m} \tm' \tmtwo'$}
		\DisplayProof
		&
		\AxiomC{$\tm \partobvind n \tm'$}		
		\AxiomC{$\val \partobvind m \val'$}
		\BinaryInfC{$(\la\var \tm)\val \partobvindlong {n + \sizep{\tm'}{\var} \cdot m +1} \tm'\isub \var {\val'}$} 
		\DisplayProof
	\end{array}$
	\def\ScoreOverhang{4pt}
\end{center}
The only difference with the usual parallel $\beta$ (defined in \refsect{lambda}) is the requirement that the argument is a value in the last rule. 
As before, the non-indexed parallel reduction $\partobv$ is simply obtained by  erasing the index, so that $\partobv \,=\, \bigcup_{n \in \nat} \partobvind{n}$. Similarly, it is easily seen that $\partobvind{0}$ is the identity relation on terms, $\tobv \,=\, \partobvind 1$ and $\partobvind n \,\subseteq\, \tobv^n$. 
Substitutivity of $\partobvind n$ is proved exactly as for $\partobind n$ (\reflemma{partobind-subs}).

\newcounter{l:tobv-subs}
\addtocounter{l:tobv-subs}{\value{lemma}}
\begin{lemma}[Substitutivity of $\partobvind n$]
\label{l:tobv-subs} 
\NoteProof{lappendix:tobv-subs}	
\label{p:tobv-subs-ind-par}  If $\tm \partobvind n \tm'$ and $\val \partobvind m \val'$, then $\tm \isub\var\val \allowbreak\partobvind{k} \tm' \isub\var{\val'}$ where $k = n + \sizep{\tm'}\var\cdot m$.
\end{lemma}

\section{Weak Call-by-Value Reduction, Essentially}
\label{sect:weak-cbv}
The \external step  we study for the \cbv $\l$-calculus is weak \cbv reduction $\tow$, which does not evaluate function bodies (the scope of $\l$-abstractions). Weak  \cbv reduction has practical importance, because it is the base of the ML/CAML
family of functional programming languages. We choose it also because it admits the natural and more general \emph{non-deterministic} presentation that follows, even if most of the literature rather presents it in a deterministic way.  
\begin{center}
\textsc{Weak \cbv reduction} \\[5pt]
$\begin{array}{c@{\ \ \sep}c@{\ \ \sep}c}  
 \AxiomC{}
 \UnaryInfC{$(\la\var \tm) \val \tow \tm\isub\var\val$}
 \DisplayProof 
 & 
  \AxiomC{$\tm \tow \tm'$}	
 \UnaryInfC{$\tm  \tmtwo \tow \tm'  \tmtwo$}
 \DisplayProof 
 &
 \AxiomC{$\tm \tow \tm'$}	
 \UnaryInfC{$\tmtwo  \tm \tow \tmtwo  \tm'$}
 \DisplayProof 
\end{array}$
\end{center}
Note that in the case of an
application there is no fixed order in the $\tow$-reduction of the left and right subterms. 
Such a non-determinism is harmless as $\tow$ satisfies a diamond-like property implying confluence, see 
\refpropp{essential-weak}{diamond} below. 
It is 
well-known that the diamond property implies uniform termination, because it implies 
that  all  maximal \sequence{s}  from a term 
have the same length. Such a further property is  known as \emph{random descent}, 
a special form of uniform termination already considered by Newman \cite{Newman} in 1942, see also van~Oostrom~\cite{DBLP:conf/rta/Oostrom07}.

 The inessential reduction $\tonotw$ and its parallel version $\partonotw$ are defined by:
\begin{center}
	 $\begin{array}{c\colspace c\colspace c\colspace ccc}
 \multicolumn{3}{c}{\textsc{$\neg$Weak reduction}}
 \\[5pt]
 	\AxiomC{$\tm \tobv \tmtwo$}
 	\UnaryInfC{$\la\var \tm \tonotw \la\var \tmtwo $}
 	\DisplayProof
 	&
 		\AxiomC{$\tm \tonotw \tm'$}	
 	\UnaryInfC{$\tm  \tmtwo \tonotw \tm'  \tmtwo$}
 	\DisplayProof 
 	&
 	\AxiomC{$\tm \tonotw \tm'$}	
 	\UnaryInfC{$\tmtwo  \tm \tonotw \tmtwo  \tm'$}
 	\DisplayProof 
 \end{array}$
\end{center}
\begin{center}
	\textsc{Parallel $\neg$weak reduction}
		\def\ScoreOverhang{1pt}
	\begin{align*}
  \AxiomC{}
    \UnaryInfC{$\var \partonotw \var$}
  \DisplayProof
	\qquad
	\AxiomC{$\tm \partobv \tm'$}	
	\UnaryInfC{$\la\var\tm \partonotw \la\var\tm'$}
	\DisplayProof 
	\qquad
	\AxiomC{$\tm \partonotw \tm'$}
	\AxiomC{$\tmtwo \partonotw \tmtwo'$}
	\BinaryInfC{$\tm \tmtwo \partonotw \tm' \tmtwo'$}
	\DisplayProof 
\end{align*}
	\def\ScoreOverhang{4pt}
\end{center}
It  is  immediate to check that $\tobv \, = \, \tow \cup \tonotw$ and
  $\tonotw \,\subseteq\, \partonotw \,\subseteq\, \tonotw^*$.

\paragraph{Weak \cbv factorization.}
We show that  $(\Lambda, \{ \tow, \tonotw \})$ is a \macrostep\ system, with $\partobv,\partonotw$ as \macrostep{s}. Merge and split are proved exactly as in \refsect{head}.

\newcounter{prop:macro-cbv}
\addtocounter{prop:macro-cbv}{\value{proposition}} 
\begin{proposition}[Weak \cbv macro-step system]
\label{prop:macro-cbv}
\NoteProof{propappendix:macro-cbv}
\begin{enumerate}
 \item\label{p:macro-cbv-merge} \emph{Merge}: if $\tm \partonotw\cdot \tow \tmthree$ then $\tm \partobv \tmthree$.
 \item\label{p:macro-cbv-indexed-split} \emph{Indexed split}: if $\tm \partobvind n \tmtwo$ then $\tm \partonotw\tmtwo$, or $n>0$ and $\tm \tow \!\cdot\! \partobvind {n-1} \tmtwo$.
 \item\label{p:macro-cbv-split} \emph{Split}: if $\tm \partobv \tmtwo$ then $\tm \tow^* \!\cdot\! \partonotw\tmtwo$.
\end{enumerate}
That is, $(\Lambda, \set{\tow, \tonotw})$ is a macro-step system with respect to $\partobv$ and $\partonotw$.
\end{proposition}

\begin{theorem}[Weak \cbv factorization]
\label{thm:weak-fact}
 If $\tm \tobv^* \tmtwo$ then $\tm\tow^* \!\cdot\! \tonotw^* \tmtwo$.
\end{theorem}

\paragraph{Plotkin's left reduction.}
 The same argument at work in this section adapts easily to factorization with respect to leftmost weak reduction (used by Plotkin \cite{DBLP:journals/tcs/Plotkin75}), or to rightmost weak reduction, the two natural deterministic variants of $\tow$.

\paragraph{Weak \cbv normalization.} 
To obtain a normalization theorem for $\tow$ via the essential normalization theorem (\refthm{par-normalization}), we need 
persistence and uniform termination. 
The latter 
{follows from} the well-known diamond property of $\tow$.

\newcounter{essential-weak}
\addtocounter{essential-weak}{\value{proposition}}
\begin{proposition}[Weak \cbv essential system]
	\label{prop:essential-weak}
	\NoteProof{propappendix:essential-weak}
	\begin{enumerate}
	 \item\label{p:essential-weak-persistence} \emph{Persistence}: if $\tm \tow \tmtwo$ and $\tm \tonotw \tmthree$ then $\tmthree \tow \tmfour$ for some $\tmfour$.
	 \item\label{p:essential-weak-diamond} \emph{Diamond}: if $\tmtwo\ltow \!\cdot\! \tow\tmthree$ with $\tmtwo \neq \tmthree$ then $\tmtwo \tow \!\cdot\! \ltow \tmthree$.
	\end{enumerate}
Then, $(\Lambda, \set{ \tow, \tonotw})$ is an essential system.
\end{proposition}

\begin{theorem}[Weak \cbv normalization]
	\label{thm:cbv-normalization}
		If $\tm \tobv^* \tmtwo$ and $\tmtwo$ is a $\tow$-normal form, then $\tm$ is strongly $\tow$-normalizing.
\end{theorem}

\cbv is often considered with respect to \emph{closed} terms only. 
In such a case the $\tow$-normal forms are exactly the (closed) values.
Then weak \cbv normalization (\refthm{cbv-normalization}) implies the following, analogous to Corollary 1 in Plotkin \cite{DBLP:journals/tcs/Plotkin75} (the  result is there obtained from standardization). 

\begin{corollary} Let $\tm$ be a closed term. 
	If $\tm \tobv^* \val$ for some value $\val$, then every maximal $\tow$-\sequence from $\tm$ is finite and ends in a value.
\end{corollary}

\section{Leftmost-Outermost Reduction, Essentially}
\label{sect:leftmost}
Here we apply our technique to leftmost-outermost (shortened to \emph{\LO}) reduction $\tolo$, the first example of \emph{\complete} 
reduction for $\tob$. 
The technical development is slightly different from the ones in the previous sections, as factorization relies on persistence. 
The same shall happen for the \complete \ll reduction of the next section. 
It seems to be a feature of \complete reductions~for~$\tob$.

\paragraph{\LO and $\neg$\LO reductions.} The definition of \lo reduction relies on two mutually recursive predicates defining normal and neutral terms (neutral = normal and not an abstraction):
\begin{center}
		\def\ScoreOverhang{2pt}
\textsc{Normal and neutral terms}\\[5pt]
{\small$\begin{array}{c@{\qquad}c@{\qquad}c@{\qquad}c@{\qquad}c}      
  \AxiomC{}
  \UnaryInfC{$\var$ is neutral}
  \DisplayProof      
  &
  \AxiomC{$\tm$ is neutral}
  \AxiomC{$\tm$ is normal}
  \BinaryInfC{$\tm\tmtwo$ is neutral}  
  \DisplayProof      
  &
  \AxiomC{$\tm$ is neutral}
  \UnaryInfC{$\tm$ is normal}
  \DisplayProof
  &
  \AxiomC{$\tm$ is normal}
  \UnaryInfC{$\la\var\tm$ is normal}
  \DisplayProof
\end{array}$ }
	\def\ScoreOverhang{4pt}
\end{center}
Dually, a term is not neutral if it is an abstraction or it is not normal. It is standard that these predicates correctly capture $\beta$ normal forms and neutrality.

The reductions of the \lo macro-step system are:
 \begin{center}
 $\begin{array}{c\colspace \colspace ccccc}
 \multicolumn{2}{c}{\textsc{\LO reduction}}
 \\[5pt]
	\AxiomC{}
	\UnaryInfC{$(\la\var \tm) \tmtwo \tolo \tm\isub\var\tmtwo$}
	\DisplayProof
	&
	\AxiomC{$\tm \tolo \tmtwo$}
	\AxiomC{$\tm\neq\la\var\tm'$}
	\BinaryInfC{$\tm \tmthree \tolo \tmtwo \tmthree$}
	\DisplayProof \\\\
	
	\AxiomC{$\tm \tolo \tmtwo$}
	\UnaryInfC{$\la\var\tm \tolo \la\var\tmtwo$}
	\DisplayProof 
	&
	  \AxiomC{$\tmthree$ is neutral}
	\AxiomC{$\tm \tolo \tmtwo$}
	\BinaryInfC{$\tmthree \tm \tolo \tmthree \tmtwo$}
	\DisplayProof 
\end{array}$
\end{center}
\begin{center}
 \begin{tabular}{c}
$\begin{array}{c\colspace \colspace ccccc}
\multicolumn{2}{c}{\textsc{$\neg$\LO reduction}}
\\[5pt]
	\AxiomC{$\tm \tob \tm'$}
	\UnaryInfC{$(\la\var \tm) \tmtwo \tonotlo (\la\var \tm') \tmtwo$}
	\DisplayProof
	&
	\AxiomC{$\tm$ is not neutral}
	\AxiomC{$\tmtwo \tob \tmtwo'$}
	\BinaryInfC{$\tm \tmtwo \tonotlo \tm \tmtwo '$}
	\DisplayProof
	\end{array}$
	\\\\
$	\begin{array}{c\colspace \colspace c\colspace \colspace cccc}
	\AxiomC{$\tm \tonotlo \tm'$}
	\UnaryInfC{$\tm \tmtwo \tonotlo \tm' \tmtwo $}
	\DisplayProof 
	&
	\AxiomC{$\tm \tonotlo \tm'$}
	\UnaryInfC{$\tmtwo \tm  \tonotlo  \tmtwo \tm'$}
	\DisplayProof 
	&	
	\AxiomC{$\tm \tonotlo \tm'$}
	\UnaryInfC{$\la\var\tm \tonotlo \la\var\tm'$}
	\DisplayProof 
\end{array}$
\end{tabular}
\end{center}
\begin{center}
\begin{tabular}{c}
	\def\ScoreOverhang{1pt}
$\begin{array}{c\colspace \colspace ccccc}
\multicolumn{2}{c}{\textsc{Parallel $\neg$\LO reduction}}
\\[5pt]
  \AxiomC{}
    \UnaryInfC{$\var \partonotlo \var$}
  \DisplayProof
	&
	\AxiomC{$\tm$ is not neutral}
	\AxiomC{$\tm \partonotlo \tm'$}
	\AxiomC{$\tmtwo \partob \tmtwo'$}	
	\TrinaryInfC{$\tm \tmtwo \partonotlo \tm' \tmtwo '$}
	\DisplayProof
\end{array}$
\\\\
\def\ScoreOverhang{1pt}
$\begin{array}{c\colspace \colspace c\colspace \colspace cccc}
	\AxiomC{$\tm \partob \tm'$}
	\AxiomC{$\tmtwo \partob \tmtwo''$}
	\BinaryInfC{$(\la\var \tm) \tmtwo \partonotlo (\la\var \tm') \tmtwo'$}
	\DisplayProof
	&
	\AxiomC{$\tm \partonotlo \tm'$}	
	\UnaryInfC{$\la\var\tm \partonotlo \la\var\tm'$}
	\DisplayProof 
	&
	\AxiomC{$\tm$ neutral}
	\AxiomC{$\tmtwo \partonotlo \tmtwo'$}
	\BinaryInfC{$\tm \tmtwo \partonotlo \tm \tmtwo'$}
	\DisplayProof 
\end{array}$
\def\ScoreOverhang{4pt}
\end{tabular}

\end{center}
As 
usual, easy inductions show that $\tob \,=\, \tolo \cup \tonotlo$ and $\tonotlo \subseteq \partonotlo \subseteq \tonotlo^*$.

Factorization depends on persistence, which is why for \lo reduction most essential properties are proved before factorization. The proofs are easy inductions.

\newcounter{prop:essential-left}
\addtocounter{prop:essential-left}{\value{proposition}}
\begin{proposition}[\lo essential properties]
\label{prop:essential-left}
\NoteProof{propappendix:essential-left}
\begin{enumerate}
\item \label{p:essential-left-completeness}
\emph{\Completeness}: if $\tm \tob \tmtwo$ then there exists $\tmthree$ such that $\tm \tolo \tmthree$.
\item\label{p:essential-left-determinism} \emph{Determinism}: if $\tm \tolo \tmtwo_1$ and $\tm\tolo \tmtwo_2$ then $\tmtwo_1 = \tmtwo_2$.
\item\label{p:essential-left-persistence} \emph{Persistence}: if $\tm \tolo \tmtwo_1$ and $\tm \tonotlo \tmtwo_2$ then $\tmtwo_2 \tolo \tmthree$ for some $\tmthree$.
\end{enumerate}
\end{proposition}

\newcounter{prop:macro-left}
\addtocounter{prop:macro-left}{\value{proposition}}
\begin{proposition}[\lo macro-step system]
\label{prop:macro-left}
\NoteProof{propappendix:macro-left}\hfill
\begin{enumerate}
	\item \emph{Merge}: \label{p:macro-left-merge}
if $\tm \partonotlo \cdot \tolo \tmthree$ then $\tm \partob \tmthree$.
	\item \label{p:macro-left-indexed-split}
	\emph{Indexed split}: if $\tm \partobind n \tmtwo$ then $\tm \partonotlo \tmtwo$, or $n>0$ and  $\tm \tolo \!\cdot\! \partobind {n-1} \tmtwo$.
	\item\label{p:macro-step-split} \emph{Split}: if $\tm \partob \tmtwo$ then $\tm \tolo^* \!\cdot\! \partonotlo \tmtwo$.
	\end{enumerate}
	That is, $(\Lambda, \set{\tolo, \tonotlo})$ is a macro-step system with respect to $\partob$ and $\partonotlo$.
\end{proposition}

\begin{proof}
We only show the merge property, and only the case that requires persistence---the rest of the proof is in \SLV{the Appendix of \cite{long}}{the Appendix}. 
The proof of the merge property is by induction on $\tm \partonotlo\tmtwo$. Consider the rule
	\[\AxiomC{$\tmfour$ not neutral}
	\AxiomC{$\tmfour \partonotlo \tmfour'$}
		\AxiomC{$\tmfive \partob \tmfive'$}
	\TrinaryInfC{$\tm = \tmfour \tmfive \partonotlo \tmfour' \tmfive ' = \tmtwo$}
	\DisplayProof\ .\]
	Since $\tmfour$ is not neutral, it is an abstraction or it is not normal. 
	If $\tmfour$ is an abstraction this case continues as the easy case of $\partonotlo$ for $\beta$-redexes (see \SLV{the Appendix of \cite{long}}{the Appendix}). 
	Otherwise, $\tmfour$ is not normal, \ie $\tmfour \tob \tmsix$. By \completeness $\tmfour \tolo \tmsix'$ for some $\tmsix'$, and by persistence (\refpropp{essential-left}{persistence}) $\tmfour'\tolo \tmfour''$ for 
	some $\tmfour''$. The hypothesis becomes $\tm = \tmfour \tmfive \partonotlo \tmfour' \tmfive ' \tolo \tmfour'' \tmfive' = \tmthree$ with $\tmfour \partonotlo \tmfour'  \tolo \tmfour''$.
	By \ih, $\tmfour \partob \tmfour''$. Then,
	\[\AxiomC{$\tmfour \partob \tmfour''$}
	\AxiomC{$\tmfive\partob \tmfive'$}
	\BinaryInfC{$\tm = \tmfour \tmfive \partob \tmfour'' \tmfive' = \tmthree$}
	\DisplayProof
	\ .
	\qedhere\]
\end{proof}

\paragraph{\lo split.} As pointed out in \refsect{Takahashi}, Takahashi's proof \cite{DBLP:journals/iandc/Takahashi95} of the split property relies on left substitutivity of head reduction, that is, if $\tm \toh \tmtwo$ then $\tm \isub\var\tmthree \toh \tmtwo \isub\var\tmthree$ for all terms $\tmthree$. Such a property does not hold for \lo reduction. For instance, $\tm = \var (I \vartwo) \tolo \var \vartwo = \tm'$ but 
$\tm\isub\var {\la\varthree \varthree\varthree} = (\la\varthree \varthree\varthree)(I\vartwo) \not\tolo (\la\varthree \varthree\varthree) \vartwo = \tm'\isub\var {\la\varthree \varthree\varthree}$.
Therefore her proof technique for factorization cannot prove the factorization theorem for \lo reduction (see also footnote \ref{f:direct}, page \pageref{f:direct}).\medskip

From \refprop{macro-left} it follows the factorization theorem for \lo reduction, that together with \refprop{essential-left} proves that $(\Lambda, \set{\tolo, \tonotlo})$ is an essential system, giving normalization of $\tolo$ for $\tob$.
\begin{theorem}
\label{thm:LO-fact-norm}
\hfill
\begin{enumerate}
\item  \emph{\lo factorization}: if $\tm \tob^* \tmu$ then $\tm\tolo^* \!\cdot\!~ \tonotlo^* \tmu$.

\item \emph{\lo normalization}: $\tolo$ is a normalizing reduction for $\tob$.
	\end{enumerate}
\end{theorem}


\section{Least-Level Reduction, Essentially}
\label{sect:stratified-cbn}
\newcommand{\sth}{\text{ s.t. }}
In this section we study another  normalizing \complete reduction for $\tob$, namely \emph{least-level} (shortened to \ll) \emph{reduction} $\tostrat$, which is non-deterministic. 
The intuition is that \ll reduction 
 fires a $\beta$-redex of minimal level, where the \emph{level} of a $\beta$-redex is the number of arguments containing it. 

The definition of $\tostrat$ relies on an indexed variant $\tostratind{k}$ of $\tob$, where $k \in \nat$ is the level of the fired $\beta$-redex (do not mix it up with the index of $\partobind n$). 
We also define a parallel version $\partostratind{n}$ (with $n \in \nat \cup \{\infty\}$) of $\tostratind{k}$, obtained as a decoration of $\partob$, where $n$ is the minimal level of the $\beta$-redexes fired by a $\partob$ step ($\partostratind{\infty}$ does not reduce any $\beta$-redex).
From now on, $\nat \cup \{\infty\}$ is considered with its usual order and arithmetic, that is, $\infty + 1 = \infty$.
\begin{center}
	\textsc{ $\beta$ reduction of level $k$}
	\def\ScoreOverhang{1pt}
	\begin{align*}
	\AxiomC{}
	\UnaryInfC{$(\la{\var}{\tm})\tmtwo \tostratind{0} \tm \isub{\var}{\tmtwo}$}
	\DisplayProof
	\qquad
	\AxiomC{$\tm \tostratind{k} \tm'$}	
	\UnaryInfC{$\la\var\tm \tostratind{k} \la\var\tm'$}
	\DisplayProof 
	\qquad
	\AxiomC{$\tm \tostratind{k} \tm'$}
	\UnaryInfC{$\tm \tmtwo \tostratind{k} \tm' \tmtwo$}
	\DisplayProof 
	\qquad
	\AxiomC{$\tm \tostratind{k} \tm'$}
	\UnaryInfC{$\tmtwo\tm \tostratind{k\!+\!1} \tmtwo\tm'$}
	\DisplayProof 
	\end{align*}
	\def\ScoreOverhang{4pt}
\end{center}
\begin{center}
	\textsc{Parallel $\beta$ reduction of least level $n$}
	\def\ScoreOverhang{1pt}
	\begin{align*}
	\AxiomC{$\tm \partostratind{k} \tm'$}
	\AxiomC{$\tmtwo \partostratind{h} \tmtwo'$}
	\BinaryInfC{$(\la{\var}{\tm})\tmtwo \partostratind{0} \tm' \isub{\var}{\tmtwo'}$}
	\DisplayProof
	\quad
	\AxiomC{$\tm \partostratind{k} \tm'$}	
	\UnaryInfC{$\la\var\tm \partostratind{k} \la\var\tm'$}
	\DisplayProof 
	\quad
	\AxiomC{$\tm \partostratind{k} \tm'$}	
	\AxiomC{$\tmtwo \partostratind{h} \tmtwo'$}
	\BinaryInfC{$\tm \tmtwo \partostratind{\min\{k,h\!+\!1\}} \tm' \tmtwo'$}
	\DisplayProof 
	\quad
	\AxiomC{}
	\UnaryInfC{$\var \partostratind{\infty} \var$}
	\DisplayProof 
	\end{align*}
	\def\ScoreOverhang{4pt}
\end{center}
Note  that $\tm \tob \tmtwo$ if and only if $\tm \tostratind{k} \tmtwo$ for some $k \in \nat$.

The \emph{least (reduction) level} $\Deg{\tm}\in \nat \cup \{\infty\}$ of a term $\tm$ is defined as follows:
\begin{align*}
\Deg{\var} &= \infty
&
\Deg{\la{\var}{\tm}} &= \Deg{\tm}
&
\Deg{\tm\tmtwo} &= 
\begin{cases}
0 &\textup{if } \tm = \la{\var}{\tmthree} 
\\
\min \{\Deg{\tm}, \Deg{\tmtwo} \!+\! 1 \} &\textup{otherwise.}
\end{cases}
\end{align*}  

The definitions of \emph{\ll}, \emph{$\neg$\ll}, and \emph{parallel $\neg$\ll reductions} are:
\begin{center}
\begin{tabular}{r\colspace c\colspace l}
 \textsc{\ll reduction} & $\tm \tostrat \tmtwo$  & if $\tm \tostratind{k} \tmtwo$ with $\deg\tm = k \in \nat$;
 \\
 \textsc{$\neg$\ll reduction} & $\tm \tonotstrat \tmtwo$ & if $\tm \tostratind{k} \tmtwo$ with $\deg\tm< k \in \nat$;
 \\
 \textsc{Parallel $\neg$\ll reduction} &$\tm \partonotstrat \tmtwo$ & if $\tm \partostratind{k} \tmtwo$ with $k \!=\! \infty$ or $k \!>\! \Deg{\tm}$.
 \end{tabular}
\end{center}
As 
usual, easy inductions show that $\tob = \toll \cup \tonotll$ and $\tonotll \subseteq \partonotll \subseteq \tonotll^*$.



\newcounter{prop:deg}
\addtocounter{prop:deg}{\value{proposition}}
\begin{proposition}[Least level properties]
\label{prop:deg}
\NoteProof{propappendix:deg}
Let $\tm$ be a term.
\begin{enumerate}
	\item\label{p:deg-finite}\emph{Computational meaning of $\llsym$}: $\Deg{\tm} = \inf\{ k \in \nat \mid \tm \tostratind{k} \tmthree \text{ for some term } \tmthree\}$.
	\item\label{p:deg-invariance-leq} \emph{Monotonicity}: if $\tm \tob \tmtwo$ then $\Deg{\tmtwo} \geq \Deg{\tm}$. 
	\item\label{p:deg-invariance-equal} \emph{Invariance by $\tonotstrat$}: if $\tm \tonotstrat \tmtwo$ then $\Deg{\tmtwo} = \Deg{\tm}$. 
\end{enumerate}
\end{proposition}

\refpoint{deg-finite} captures the meaning of the least level, and gives \completeness of $\toll$. In particular, $\Deg{\tm} = \infty$ if and only if $\tm$ is $\tob$-normal, since $\inf \, \emptyset = \infty$. 
Monotonicity states that $\beta$ steps cannot decrease the least level. 
Invariance by $\tonotstrat$ 
says that $\tonotstrat$ cannot change the least level. Essentially, this is persistence. 
		
\newcounter{prop:essential-strat}
\addtocounter{prop:essential-strat}{\value{proposition}}
\begin{proposition}[\ll essential properties]
\label{prop:essential-strat}
\NoteProof{propappendix:essential-strat}\hfill
\begin{enumerate}
\item \label{p:essential-strat-completeness}
\emph{\Completeness}: if $\tm \tob \tmtwo$ then $\tm \toll \tmthree$ for some $\tmthree$.
\item\label{p:essential-strat-persistence} \emph{Persistence}: if $\tm \toll \tmtwo_1$ and $\tm \tonotll \tmtwo_2$ then $\tmtwo_2 \toll \tmthree$ for some $\tmthree$.
\item\label{essential-strat-diamond} \emph{Diamond}: if $\tmtwo\ltoll \cdot \toll\tmthree$ with $\tmtwo \neq \tmthree$ then $\tmtwo \toll \cdot \ltoll \tmthree$.
\end{enumerate}
\end{proposition}

%

As for \lo, merge needs persistence, or, more precisely, invariance by $\tonotstrat$.

\newcounter{prop:macro-strat}
\addtocounter{prop:macro-strat}{\value{proposition}}
\begin{proposition}[\ll macro-step system]
\hfill\label{prop:macro-strat}
\NoteProof{propappendix:macro-strat}
	\begin{enumerate}
		\item\label{p:macro-strat-merge} \emph{Merge}: if $\tm \partonotstrat \tmtwo \tostrat \tmthree$, then $\tm \partob \tmthree$.
		\item\label{p:macro-strat-indexed-split}\emph{Indexed split}:	if $\tm \partobind n \tmtwo$ then $\tm \partonotstrat\tmtwo$, or $n>0$ and $\tm \tostrat \cdot \partobindlong {n-1} \tmtwo$.
		\item\label{p:macro-strat-split} \emph{Split}: if $\tm \partob \tmtwo$ then $\tm \tost^* \cdot \partonotstrat \tmtwo$.
	\end{enumerate}
	That is, $(\Lambda, \set{\toll, \tonotll})$ is a macro-step system with respect to $\partob$ and $\partonotll$.
\end{proposition}

\begin{theorem}
\hfill
\begin{enumerate}
\item  \emph{\ll factorization}: if $\tm \tob^* \tmu$ then $\tm\toll^* \cdot \tonotll^* \tmu$.

\item \emph{\ll normalization}: $\toll$ is a normalizing reduction for $\tob$.
	\end{enumerate}
\end{theorem}

\paragraph{\ll split and \lo \textit{vs.} \ll.} As for \lo reduction, left substitutivity does not hold for $\toll$. 
Consider $\tm = \var (I \vartwo) \toll \var \vartwo = \tm'$ where the step has level 1, and $\tm\isub\var {\la\varthree \varthree\varthree} \allowbreak= (\la\varthree \varthree\varthree)(I\vartwo) \not\toll (\la\varthree \varthree\varthree) \vartwo = \tm'\isub\var {\la\varthree \varthree\varthree}$ since now there also is a step $(\la\varthree \varthree\varthree)(I\vartwo)\toll (I \vartwo)(I \vartwo)$ at level 0.

Moreover, \ll and \lo reductions are incomparable. First, note that
	\emph{$\toll \not\subseteq 
	\tolo$}: $\tm = (\la{\var}{II})\vartwo  \tostrat (\la{\var}{I})\vartwo = \tmtwo$, because $\tm \tostratind{0} (\la{\var}{I})\vartwo$ and $\Deg{\tm} = 0$, but $\tm \not\tolo \tmtwo$, indeed $\tm \tolo II$. This fact also shows that $\toll$ is not left–outer in the sense of van Oostrom and Toyama \cite{DBLP:conf/rta/OostromT16}.
	Second, \emph{
	$\tolo \not\subseteq \toll$}: $\tm = \var (\var(II)) (II) \tolo \var (\var I)(II) = \tmtwo$ but $\tm \not\tostrat \tmtwo$, indeed $\tm \tonotstrat \tmtwo$ because $\tm \tostratind{2} \tmtwo$ and $\Deg{\tm} = 1$, and $\tm \toll \var (\var(II)) I \neq \tmtwo$.


\section{Conclusions} 
We provide simple proof techniques for factorization and normalization theorems in the $\l$-calculus, simplifying Takahashi's parallel method \cite{DBLP:journals/iandc/Takahashi95}, 
extending its scope and making it more abstract at the same time. 

About the use of parallel reduction, 
Takahashi	 claims: ``\textit{once the idea is stated properly, the essential part of the proof is almost over, because the inductive verification of the statement is easy, even mechanical}'' \cite[p.~122]{DBLP:journals/iandc/Takahashi95}. 
Our work reinforces this point of view, as our case studies smoothly follow the abstract~schema.

%

\paragraph{Range of application.} 
We apply our method for factorization and
	normalization to two notions of reductions that compute full normal forms:
\begin{itemize}
	\item the classic example of \lo reduction, 
		covered also by the recent techniques by Hirokawa, Middeldorp and Moser \cite{DBLP:conf/rta/HirokawaMM15} and 
	van Oostrom and Toyama \cite{DBLP:conf/rta/OostromT16}; 
	\item 
	\ll reduction, which is out of the scope of \cite{DBLP:conf/rta/HirokawaMM15,DBLP:conf/rta/OostromT16} because it is neither deterministic (as required by \cite{DBLP:conf/rta/HirokawaMM15}), nor  left–outer in the sense of \cite{DBLP:conf/rta/OostromT16} (as pointed out here in \refsect{stratified-cbn}). 
\end{itemize}

Our approach naturally covers also reductions that do not compute full normal forms, such as head and weak \cbv reductions.
	These results are out of reach for 
van Oostrom and Toyama's technique \cite{DBLP:conf/rta/OostromT16}, as they clarify in their conclusions.

Because of the minimality of our assumptions, we believe that our method applies to a large variety of other cases and variants of the $\l$-calculus. 
\SLV{}{
A key feature of our approach is that it derives normalization from factorization.
However, it is worth noting that factorization is not a necessary condition for normalization.\footnotemark
\footnotetext{For instance, in the weak $\l$-calculus---where weak $\beta$-reduction $\towb$ does not reduce under abstractions---our technique does not apply because weak head reduction $\towh$ (\ie head reduction that does not reduce under abstractions) satisfies a weak head normalization theorem (if $\tm \towb^* \tmtwo$ with $\tmtwo$ $\towh$-normal then $\towh$ terminates on $\tm$) but does not factorize: indeed, given the $\towb$-sequence $(\la{\var}\la{\vartwo}\var)(II) \tonotwh (\la{\var}\la{\vartwo}\var)I \towh \la{\vartwo}I$, there is no term $\tm$ such that $(\la{\var}\la{\vartwo}\var)(II) \towh^* \tm \tonotwh^* \la{\vartwo}I$.}
}


\paragraph{Acknowledgments.}
This work has been partially funded by the ANR JCJC grant COCA HOLA (ANR-16-CE40-004-01) and by the EPSRC grant EP/R029121/1
``Typed Lambda-Calculi with Sharing and Unsharing''.

\bibliographystyle{splncs04}
\bibliography{biblio-short}
\addcontentsline{toc}{section}{References}

 \newpage
 \appendix

\section{Technical appendix: omitted proofs and lemmas}
\label{sect:proofs}

The enumeration of propositions, theorems, lemmas already stated in the body of the article is unchanged.

\subsection{Omitted proofs  of \refsect{head} (head)}

\setcounter{propositionAppendix}{\value{prop:macro-head}}
\begin{propositionAppendix}[Head macro-step system]
	\label{propappendix:macro-head}
	\NoteState{prop:macro-head}
	\hfill
	\begin{enumerate}
		\item\label{pappendix:macro-head-merge} \emph{Merge}: if $\tm \partonoth\cdot \toh \tmthree$ then $\tm \partob \tmthree$.
		\item\label{pappendix:macro-head-indexed-split} \emph{Indexed split}: if $\tm \partobind n \tmtwo$ then $\tm \partonoth\tmtwo$, or $n>0$ and $\tm \toh \cdot \partobind {n-1} \tmtwo$.
		
		\item\label{pappendix:macro-head-split} \emph{Split}: if $\tm \partob \tmtwo$ then $\tm \toh^* \cdot \partonoth\tmtwo$.
	\end{enumerate}
	that is, $(\Lambda, \{ \toh, \tonoth \})$ is a \macrostep\ system with respect to $\partob,\partonoth$.
\end{propositionAppendix}

\begin{proof}
	\refpoints{macro-head-indexed-split}{macro-head-split} are already proved on p.~\pageref{prop:macro-head}.
	We prove \refpoint{macro-head-merge} by induction on the definition of $\tm \partonoth \tmtwo$.
	Cases:
	\begin{itemize}
		\item \emph{Variable}: $\tm = \var \partonoth \var = \tmtwo$. 
		Then, there is no $\tmthree$ such that $\tmtwo \toh \tmthree$.
		
		\item \emph{Abstraction}: $\tm = \la \var\tm' \partonoth  \la \var \tmtwo' = \tmtwo$ because $\tm' \partonoth \tmtwo'$.
		According to the definition of $\tmtwo \toh \tmthree$, necessarily $\tmthree = \la{\var}{\tmthree'}$ with $\tmtwo' \toh \tmthree'$.
		By \ih applied to $\tm'$, we have $\tm' \partob \tmthree'$ and hence:
		\begin{align*}
			\AxiomC{$\tm' \partob \tmthree'$}
			\UnaryInfC{$\tm = \la{\var}{\tm'} \partob \la{\var}{\tmthree'} = \tmthree$}
			\DisplayProof
			\,.
		\end{align*}
	
		\item \emph{Application}:
		\begin{align}
		\label{eq:app}
		\AxiomC{$\tmfour \partonoth \tmfour'$}
		\AxiomC{$\tmfive \partob \tmfive'$}
		\BinaryInfC{$\tm = \tmfour\tmfive \partonoth \tmfour'\tmfive' = \tmtwo$}
		\DisplayProof
		\end{align}	
		Sub-cases:
		\begin{enumerate}
			\item $\tmtwo = \tmfour'\tmfive' \toh \tmfour''\tmfive' = \tmthree$ with $\tmfour' \toh \tmfour''$; 
			by \ih applied to $\tmfour \partonoth \tmfour' \toh \tmfour''$, we have $\tmfour \partob \tmfour''$, and so (as $\partonoth \,\subseteq\, \partob$)
			\begin{align*}
			\AxiomC{$\tmfour \partob \tmfour''$}
			\AxiomC{$\tmfive \partob \tmfive'$}
			\BinaryInfC{$\tm = \tmfour\tmfive \partob \tmfour''\tmfive' = \tmthree$}
			\DisplayProof\,;
			\end{align*}
			
			\item $\tmtwo = (\la{\var}\tmsix')\tmfive' \toh \tmsix'\isub{\var}{\tmfive'} = \tmthree$, which means that $\tmfour' = \la{\var}{\tmsix'}$ in \eqref{eq:app}.
			According to the definition of $\tmfour \partonoth \tmfour'$, necessarily $\tmfour = \la{\var}{\tmsix}$ with $\tmsix \partonoth \tmsix'$.
			Thus, (as $\partonoth \,\subseteq\, \partob$)
			\begin{align*}
			\AxiomC{$\tmq\partob \tmq'$}
			\AxiomC{$\tmfive \partob \tmfive'$}
			\BinaryInfC{$\tm =  (\la {\var} \tmsix) \tmfive \partob  \tmq'\isub{\var}{\tmfive'} = \tmthree$}
			\DisplayProof
			\,.
			\end{align*}
		\end{enumerate}
		\item \emph{$\beta$-redex}:
		\begin{align*}
			\AxiomC{$\tmfour \partob \tmfour'$}
			\AxiomC{$\tmfive \partob \tmfive'$}
			\BinaryInfC{$\tm = (\la\var \tmfour) \tmfive \partonoth (\la\var \tmfour') \tmfive' = \tmtwo$}
			\DisplayProof
		\end{align*}
		According to the definition of $\tmtwo \toh \tmthree$, we have $\tmthree = \tmfour'\isub{\var}{\tmfive'}$.
		Hence,
		\[
			\AxiomC{$\tmfour \partob \tmfour'$}
			\AxiomC{$\tmfive \partob \tmfive'$}
			\BinaryInfC{$\tm = (\la\var \tmfour) \tmfive \partob \tmfour'\isub{\var}{\tmfive'} = \tmthree$}
			\DisplayProof
			\ . \qedhere
		\]
	\end{itemize}
\end{proof}

\setcounter{propositionAppendix}{\value{prop:essential-head}}
\begin{propositionAppendix}[Head essential system]
	\label{propappendix:essential-head}
\NoteState{prop:essential-head}
\begin{enumerate}
 \item\label{pappendix:essential-head-persistence} \emph{Persistence}: if $\tm \toh \tmtwo$ and $\tm \tonoth \tmthree$ then $\tmthree \toh \tmfour$ for some $\tmfour$.
 \item\label{pappendix:essential-head-determinism} \emph{Determinism}: if $\tm \toh \tmtwo_1$ and $\tm\toh \tmtwo_2$ then $\tmtwo_1 = \tmtwo_2$.
\end{enumerate}
Then, $(\Lambda, \set{ \tow, \tonotw})$ is an essential system.
\end{propositionAppendix}

\begin{proof}\hfill
\begin{enumerate}
 \item \emph{Persistence}: By induction on $\tm \toh \tmtwo$. Cases:
 \begin{itemize}
  \item \emph{Root step}, \ie $\tm = (\la\var \tmfive) \tmsix \toh \tmfive\isub\var\tmsix = \tmtwo$. Sub-cases:
  \begin{itemize}
    \item \emph{$\tonoth$ in the left sub-term}: $\tm = (\la\var \tmfive) \tmsix \tonoth (\la\var \tmfive') \tmsix=\tmthree$ because $\tmfive \tob \tmfive'$ or $\tmfive \tonoth \tmfive'$. In both cases $\tmthree = (\la\var \tmfive') \tmsix \toh \tmfive'\isub\var\tmsix =:\tmfour$.
    
    \item \emph{$\tonoth$ in the right sub-term}: $\tm = (\la\var \tmfive) \tmsix \tonoth (\la\var \tmfive) \tmsix'=\tmthree$ because $\tmsix \tob \tmsix'$. Then $\tmthree = (\la\var \tmfive) \tmsix' \toh \tmfive\isub\var{\tmsix'} =:\tmfour$.    
  \end{itemize}
  \item \emph{Application}: \ie $\tm = \tmfive \tmsix \toh \tmfive'\tmsix = \tmtwo$ with $\tmfive \toh \tmfive'$ and $\tmfive$ not an abstraction. Sub-cases:
  \begin{itemize}
    \item \emph{$\tonoth$ in the left sub-term}: $\tm = \tmfive \tmsix \tonoth \tmfive'' \tmsix=\tmthree$ because $\tmfive \tonoth \tmfive''$. By \ih, $\tmfive'' \toh \tmfour'$ for some $\tmfour'$. Then $\tmthree = \tmfive'' \tmsix \toh \tmfour'\tmsix =:\tmfour$.
    
    \item \emph{$\tonoth$ in the right sub-term}: $\tm = \tmfive \tmsix \tonoth \tmfive \tmsix'=\tmthree$ because $\tmsix \tob \tmsix'$. Then $\tmthree = \tmfive \tmsix' \toh \tmfive'\tmsix' =:\tmfour$.   
  \end{itemize}
  \item \emph{Abstraction}: \ie $\tm = \la\var \tmfive \toh \tmfive' = \tmtwo$. Then $\tm = \la\var \tmfive \tonoth \la\var\tmfive'' = \tmthree$ for some $\tmfive''$ with $\tmfive \tonoth \tmfive''$. By \ih, $\tmfive'' \toh \tmfour'$ for some $\tmfour'$. Then $\tmthree = \la\var\tmfive'' \toh \la\var\tmfour'$.
 \end{itemize}


	\item \emph{Determinism}: By induction on a derivation with conclusion $\tm \toh \tmtwo_1$. 
	Consider its last rule.
	Cases:
	\begin{itemize}
		\item \emph{Abstraction}, \ie $\tm = \la{\var}{\tm'} \toh \la{\var}{\tmtwo_1'} = \tmtwo_1$ because $\tm' \toh \tmtwo_1'$.
		According to the definition of $\toh$, the only possibility for the last rule of a derivation for $\tm \toh \tmtwo_2$ is 
		\begin{align*}
			\AxiomC{$\tm' \toh \tmtwo_2'$}
			\UnaryInfC{$\tm = \la{\var}{\tm'} \toh \la{\var}{\tmtwo_2'} = \tmtwo_2$}
			\DisplayProof
			\ .
		\end{align*}
		By \ih applied to $\tm'$, we have $\tmtwo_1' = \tmtwo_2'$ and hence $\tmtwo_1 = \la{\var}{\tmtwo_1'} = \la{\var}{\tmtwo_2'} = \tmtwo_2$.
		
		\item \emph{Application}, \ie $\tm = \tm' \tm'' \toh \tmtwo_1'\tm'' = \tmtwo_1$ because $\tm' \toh \tmtwo_1'$ and $\tm' \neq \la{\var}{\tmfour}$.
		According to the definition of $\toh$, the the last rule of a derivation for $\tm \toh \tmtwo_2$ can be only 
		\begin{align*}
		\AxiomC{$\tm' \toh \tmtwo_2'$}
		\UnaryInfC{$\tm = \tm'\tm'' \toh \tmtwo_2'\tm'' = \tmtwo_2$}
		\DisplayProof
		\ .
		\end{align*}
		By \ih applied to $\tm'$, we have $\tmtwo_1' = \tmtwo_2'$ and hence $\tmtwo_1 = \tmtwo_1'\tm'' = \tmtwo_2'\tm'' = \tmtwo_2$.
		
		\item \emph{$\beta$-rule}, \ie $\tm = (\la{\var}\tm')\tm'' \toh \tm'\isub{\var}{\tm''} = \tmtwo_1$.
		According to the definition of $\toh$, the only possibility for the last rule of a derivation for $\tm \toh \tmtwo_2$ is 
		\begin{align*}
			\AxiomC{}
			\UnaryInfC{$\tm = (\la{\var}{\tm'})\tm'' \toh \tm'\isub{\var}{\tm''} = \tmtwo_2$}
			\DisplayProof
			&&\text{and hence } \tmtwo_1 = \tm'\isub{\var}{\tm''} = \tmtwo_2.
			\text{\qedhere}
		\end{align*}
	\end{itemize}
\end{enumerate}
 
\end{proof}

%
%
%
%
%
%
%
%
%

\subsection{Omitted proofs of \refsect{cbv-calculus} (\texorpdfstring{\cbv $\l$-calculus}{\cbv lambda-calculus})}
\setcounter{lemmaAppendix}{\value{l:tobv-subs}}
\begin{lemmaAppendix}[Substitutivity of $\partobvind n$]
	\label{lappendix:tobv-subs} 
	\NoteState{l:tobv-subs}
 If $\tm \partobvind n \tm'$ and $\val \partobvind m \val'$ then $\tm \isub\var\val \partobvind{k} \tm' \isub\var{\val'}$ where $k = n + \sizep{\tm'}\var\cdot m$.
\end{lemmaAppendix}

\begin{proof}
By induction on $\tm \partobind n \tm'$. It follows the exact same pattern of the proof of substitutivity of $\partobind n$. Cases:
\begin{itemize}
	\item \emph{Variable}: two sub-cases
	\begin{itemize}
	\item $\tm = \var$: then $\tm = \var \partobvind 0 \var = \tm'$ then $\tm \isub\var\val = \var\isub\var\val = \val \partobvind m \val' = \var \isub\var{\val'} = \tm' \isub\var{\val'}$ that satisfies the statement because $n + \sizep{\tm'}\var\cdot m = 0+1\cdot m = m$. 
	 \item $\tm = \vartwo$: then $\tm = \vartwo \partobvind 0 \vartwo = \tm'$ then $\tm \isub\var\val = \vartwo\isub\var\val = \vartwo \partobvind 0 \vartwo = \vartwo \isub\var{\val'} = \tm' \isub\var{\val'}$ that satisfies the statement because $n + \sizep{\tm'}\var\cdot m = n +0\cdot m=n$.
	 \end{itemize}
	
	\item \emph{Abstraction}, \ie $\tm = \la{\vartwo}{\tmthree} \partobvind{n} \la{\vartwo}{\tmthree'} = \tm'$ because $\tmthree \partobvind{n} \tmthree'$; we can suppose without loss of generality that $\vartwo \neq \var$, hence $\sizep{\tmthree'}{\var} = \sizep{\tm'}{\var}$ and $\tm\isub{\var}{\val} = \la{\vartwo}(\tmthree\isub{\var}{\val})$ and $\tm'\isub{\var}{\val'} = \la{\vartwo}(\tmthree'\isub{\var}{\val'})$.
		By \ih, $\tmthree\isub{\var}{\val} \partobvindlong{n +  \sizep{\tmthree'}{\var} \cdot m} \tmthree'\isub{\var}{\val'}$, thus
		\begin{align*}
		\AxiomC{$\tmthree\isub{\var}{\val} \partobvindlong{n +  \sizep{\tmthree'}{\var} \cdot m} \tmthree'\isub{\var}{\val'}$}
		\UnaryInfC{$\tm\isub{\var}{\val} = \la {\vartwo}\tmthree\isub{\var}{\val} \partobvindlong{n +  \sizep{\tm'}{\var} \cdot m} \la{\vartwo}\tmthree'\isub{\var}{\val'} = \tm'\isub{\var}{\val'}$}
		\DisplayProof\,.
		\end{align*}

	\item \emph{Application}, \ie	
	\[\AxiomC{$\tmthree \partobvind {n_\tmthree} \tmthree'$}
	\AxiomC{$\tmfour \partobvind {n_\tmfour} \tmfour'$}
	\BinaryInfC{$\tm = \tmthree \tmfour \partobvindlong {n_\tmthree + n_\tmfour} \tmthree' \tmfour'= \tm'$}
	\DisplayProof\]
	with $n = n_\tmthree + n_\tmfour$. By \ih, $\tmthree\isub\var\val \partobvindlong {n_\tmthree +\sizep{\tmthree'}\var \cdot m} \tmthree'\isub\var{\val'}$ and $\tmfour\isub\var\val \partobvindlong {n_\tmfour +\sizep{\tmfour'}\var \cdot m} \tmfour'\isub\var{\val'}$. Then 
	\[\AxiomC{$\tmthree\isub\var\val \partobvindlong {n_\tmthree +\sizep{\tmthree'}\var \cdot m} \tmthree'\isub\var{\val'}$}
	\AxiomC{$\tmfour\isub\var\val \partobvindlong {n_\tmfour +\sizep{\tmfour'}\var \cdot m} \tmfour'\isub\var{\val'}$}
	\BinaryInfC{$\tm\isub\var\val = \tmthree\isub\var\val \tmfour\isub\var\val \partobvind{k} \tmthree'\isub\var{\val'} \tmfour'\isub\var{\val'}= \tm'\isub\var{\val'}$}
	\DisplayProof\]
	with $k = n_\tmthree +\sizep{\tmthree'}\var \cdot m +n_\tmfour +\sizep{\tmfour'}\var \cdot m$, which proves the statement because $$k = n_\tmthree +\sizep{\tmthree'}\var \cdot m +n_\tmfour +\sizep{\tmfour'}\var \cdot m = n + (\sizep{\tmthree'}\var + \sizep{\tmfour'}\var)\cdot m = n + \sizep{\tm'}\var \cdot m.$$

\item \emph{$\betav$-step}, \ie ($\wal$ and $\waltwo$ are values)
	\[\AxiomC{$\tmthree \partobvind {n_\tmthree} \tmthree'$}
	\AxiomC{$\wal \partobvind {n_{\wal}} \waltwo$}
\BinaryInfC{$\tm = (\la\vartwo \tmthree)\wal \partobvindlong {n_\tmthree + \sizep{\tmthree'}\vartwo \cdot n_{\wal} +1} \tmthree'\isub \vartwo {\waltwo} = \tm'$}
\DisplayProof\]
with $n = n_\tmthree + \sizep{\tmthree'}\vartwo \cdot n_{\wal} +1$. We can assume without loss of generality that $\vartwo\neq \var$, 
		hence, $\sizep{\tm'}{\var} = \sizep{\tmthree'\isub\vartwo {\waltwo}}\var = \sizep{\tmthree'}{\var} + \sizep{\tmthree'}{\vartwo}\cdot\sizep{\waltwo}{\var} $ and $\tm\isub\var\val = (\la\vartwo\tmthree\isub\var\val) (\wal\isub\var\val)$ and $\tm'\isub\var{\val'} = \tmthree'\isub\var{\val'} \isub\vartwo{\waltwo\isub\var{\val'}}$.
		
		By \ih, $\tmthree\isub\var\val \partobvindlong {n_\tmthree +\sizep{\tmthree'}\var \cdot m} \tmthree'\isub\var{\val'}$ and $\wal\isub\var\val \partobvindlong {n_{\wal} +\sizep{\waltwo}\var \cdot m} \waltwo\isub\var{\val'}$. 
		Then 
	\[
	\def\ScoreOverhang{1pt}
	\AxiomC{$\tmthree\isub\var\val \partobvindlong {n_\tmthree +\sizep{\tmthree'}\var \cdot m} \tmthree'\isub\var{\val'}$}
	\AxiomC{$\wal\isub\var\val \partobvindlong {n_{\wal} +\sizep{\waltwo}\var \cdot m} \waltwo\isub\var{\val'}$}
	\BinaryInfC{$\tm\isub\var\val = (\la\vartwo\tmthree\isub\var\val) (\wal\isub\var\val) \partobvind{k} \tmthree'\isub\var{\val'} \isub\vartwo{\waltwo\isub\var{\val'}}= \tm'\isub\var{\val'}$}
	\DisplayProof
	\def\ScoreOverhang{4pt}
	\]
where $k = n_\tmthree +\sizep{\tmthree'}\var \cdot m +\sizep{\tmthree'}\vartwo\cdot(n_{\wal} +\sizep{\waltwo}\var \cdot m)+1 = n_\tmthree +\sizep{\tmthree'}\var \cdot m +\sizep{\tmthree'}\vartwo\cdot(n_{\wal} +\sizep{\waltwo}\var \cdot m)+1 
=
n_\tmthree + \sizep{\tmthree'}\vartwo\cdot n_{\wal} + 1 + \sizep{\tmthree'}\var \cdot m + \sizep{\tmthree'}\vartwo\cdot \sizep{\waltwo}\var \cdot m 
= n + \sizep{\tm'}\var\cdot m$.
\qedhere
\end{itemize}
\end{proof}

\subsection{Omitted proofs  of \refsect{weak-cbv} (weak \cbv)}

\setcounter{propositionAppendix}{\value{prop:macro-cbv}} 
\begin{propositionAppendix}[Weak \cbv macro-step system]
	\label{propappendix:macro-cbv}
	\NoteState{prop:macro-cbv}
\begin{enumerate}
 \item \label{pappendix:macro-cbv-merge} \emph{Merge}: if $\tm \partonotw\cdot \tow \tmthree$ then $\tm \partobv \tmthree$.
 \item \label{pappendix:macro-cbv-indexed-split} \emph{Indexed split}: if $\tm \partobvind n \tmtwo$ then $\tm \partonotw\tmtwo$, or $n>0$ and $\tm \tow \cdot \partobvind {n-1} \tmtwo$.
 \item \label{pappendix:macro-cbv-split} \emph{Split}: if $\tm \partobv \tmtwo$ then $\tm \tow^* \cdot \partonotw\tmtwo$.
\end{enumerate}
That is, $(\Lambda, \set{\tow, \tonotw})$ is a macro-step system with respect to $\partobv$ and $\partonotw$.
\end{propositionAppendix}

\begin{proof}
\hfill
\begin{enumerate}
 \item \emph{Merge}: by induction on $\tm \partonotw\tmtwo$. Note that the cases in which $\tmtwo=\var$ or $\tmtwo=\l x. \tms'$ are not possible.  Hence $ \tmtwo=\tmfour'\tmfive' $  and $\tm \partonotw\tmtwo$ is derived as follows
	\begin{align*}
	\AxiomC{$\tmfour \partonotw \tmfour'$}
	\AxiomC{$\tmfive \partonotw \tmfive'$}
	\BinaryInfC{$\tm = \tmfour\tmfive \partonotw \tmfour'\tmfive' =\tmtwo$}
	\DisplayProof
	\end{align*}
\begin{enumerate}
\item If $ \tmfour' \tow \tmfour''$ then $\tmtwo = \tmfour'' \tmfive'$. The \ih gives $\tmfour \partobv \tmfour''$, and $\tm \partobv \tmtwo$ is derived as follows (remember that $\partonotw \subseteq \partobv$):
\begin{align*}
	\AxiomC{$\tmfour \partobv \tmfour''$}
	\AxiomC{$\tmfive \partobv \tmfive'$}
	\BinaryInfC{$\tm = \tmfour\tmfive \partobv \tmfour''\tmfive' =\tmtwo$}
	\DisplayProof
	\end{align*}

	\item If $ \tmfive' \tow \tmfive''$ it is analogous to the previous case.

\item If 	$ \tmtwo \tow \tmthree$ by a top $\betav$ step then $\tmfour' = \la x \tmq'$. Now, by definition of $\partonotw$ the step $\tmfour \partonotw \tmfour'$ necessarily has the form $\tmfour = \la x \tmq \partonotw \la x \tmq' = \tmfour'$ for some $\tmq$ such that $\tmq\partobv \tmq'$. Then the hypothesis is $\tm = (\la x \tmq) \tmfive \partonotw (\la x \tmq') \tmfive' \tow \tmq'\isub{\var}{\tmfive'} = \tmthree$ and  $\tm \partobv \tmtwo$ is derived as follows (remember that $\partonotw \subseteq \partobv$):
\begin{align*}
\AxiomC{$\tmq\partobv \tmq'$}
\AxiomC{$\tmfive \partobv \tmfive'$}
\BinaryInfC{$\tm = (\la x \tmq) \tmfive \partobv  \tmq'\isub{\var}{\tmfive'} =\tmtwo $}
\DisplayProof
\end{align*}
\end{enumerate}

\item \emph{Indexed split}: by induction on the definition of $\tm \partobvind n \tmtwo$. 
	We freely use the fact that if $\tm \partobvind n \tmtwo$ then $\tm \partobv \tmtwo$. Cases:
	\begin{itemize}
		\item \emph{Variable}: $\tm = \var \partobvind 0 \var = \tmtwo$. 
		Then, $\tm = \var \partonotw\var = \tmtwo$.
		
		\item \emph{Abstraction}: $\tm = \la \var\tm' \partobvind n  \la \var \tmtwo' = \tmtwo$ because $\tm' \partobvind n \tmtwo'$. 
		Then, 
		\begin{align*}
		\AxiomC{$\tm' \partobv \tmtwo'$}
		\UnaryInfC{$\tm = \la{\var}\tm' \partonotw \la{\var}\tmtwo' = \tmtwo$}
		\DisplayProof
		\,.
		\end{align*}
		
	\item \emph{Application}:
	\[\AxiomC{$\tmfour \partobvind {n_1} \tmfour'$}
	\AxiomC{$\tmfive \partobvind {n_2} \tmfive'$}
	\BinaryInfC{$\tm = \tmfour \tmfive \partobvindlong {n_1 + n_2} \tmfour' \tmfive' = \tmtwo$}
	\DisplayProof\]
	with $n = n_1 + n_2$. There are only two cases:
	\begin{itemize}
		\item either $\tmfour \tmfive \partonotw \tmfour' \tmfive'$, and then the claim holds.
		\item or $\tmfour \tmfive \not\partonotw \tmfour' \tmfive'$, and then $\tmfour \not\partonotw \tmfour'$ or $\tmfive \not\partonotw \tmfive'$ (otherwise $\tmfour \tmfive \partonotw \tmfour' \tmfive'$). 
		Suppose $\tmfour \not\partonotw \tmfour'$ (the other case is analogous). 
		By \ih applied to $\tmfour \partobvind{n_1} \tmfour'$,
		$n_1>0$ and there is $\tmfour''$ such that $\tmfour \tow \tmfour'' \partobvind {n_1-1} \tmfour'$. 
		So, $\tm = \tmfour \tmfive \tow \tmfour'' \tmfive$ and 
		\[\AxiomC{$\tmfour'' \partobvindlong {n_1-1} \tmfour'$}			
		\AxiomC{$\tmfive \partobvind {n_2} \tmfive'$}	
		\BinaryInfC{$\tmfour'' \tmfive \partobvindlong {n_1-1 + n_2} \tmfour' \tmfive' = \tmtwo$}
		\DisplayProof
		\ .
		\]
	\end{itemize}
		
		\item \emph{$\betav$ step}:
		\[\AxiomC{$\tmfive \partobvind {n_1} \tmfive'$}
		\AxiomC{$\tmfour$ is a value}
		\AxiomC{$\tmfour \partobvind {n_2} \tmfour'$}		
		\TrinaryInfC{$\tm = (\la\var \tmfive)\tmfour \partobvindlong {n_1 + \sizep{\tmfive'}\var \cdot n_2 +1} \tmfive'\isub \var {\tmfour'} = \tmtwo$}
		\DisplayProof\]
		with $n = n_1 + \sizep{\tmfive'}\var \cdot n_2 +1 > 0$. 
		We have $\tm = (\la\var \tmfive)\tmfour \tow \tmfive\isub\var\tmfour \eqdef \tmthree$ and substitutivity of $\partobvind n$ (\reflemma{tobv-subs}) gives $\tmthree = \tmfive\isub\var\tmfour \partobvindlong{n_1 + \sizep{\tmfive'}\var \cdot n_2} \tmfive' \isub\var{\tmfour'} = \tmtwo$.		
	\end{itemize}
	
	\item \emph{Split}: exactly as in the head case (\refpropp{macro-head}{split}), using the Indexed Split property for weak \cbv (\refpoint{macro-cbv-indexed-split} above).
	\qedhere
	\end{enumerate}
\end{proof}

\setcounter{propositionAppendix}{\value{essential-weak}}
\begin{propositionAppendix}[Weak \cbv essential system]
	\label{propappendix:essential-weak}
	\NoteState{prop:essential-weak}
	\begin{enumerate}
		\item\label{pappendix:essential-weak-persistence} \emph{Persistence}: if $\tm \tow \tmtwo$ and $\tm \tonotw \tmthree$ then $\tmthree \tow \tmfour$ for some $\tmfour$.
		\item\label{pappendix:essential-weak-diamond} \emph{Diamond}: if $\tmtwo\ltow \tm \tow\tmthree$ with $\tmtwo \neq \tmthree$ then $\tmtwo \tow \tmfour \ltow \tmthree$ for some $\tmfour$.
	\end{enumerate}
	Then, $(\Lambda, \set{ \tow, \tonotw})$ is an essential system.
\end{propositionAppendix}

\begin{proof}\hfill
	\begin{enumerate}
		\item \emph{Persistence}: By induciton on the definition of $\tm \tonotw \tmtwo$.
		Cases:
		\begin{itemize}
			\item \emph{Abstraction}: $\tm = \la\var\tmfive \tonotw \la\var\tmfive'$ because $\tmfive \tobv \tmfive'$. This case is impossible because $\tm$ is $\tow$ normal, against the hypothesis that $\tm \tow \tmthree$.
			
			\item \emph{Application left}: $\tm = \tmfive\tmsix \tonotw \tmfive_{\neg\wsym}\tmsix = \tmtwo$ because $\tmfive \tonotw \tmfive_{\neg\wsym}$.
			According to the definition of $\tow $, there are the following sub-cases:
			\begin{enumerate}
				\item $\tm = \tmfive\tmsix \tow \tmfive_\wsym\tmsix = \tmthree$ with $\tmfive \tow \tmfive_\wsym$; 
				by \ih applied to $\tmfive$, we have $\tmfive_{\neg\wsym} \tow \tmfour'$  for some term $\tmfour'$, and so 
				\begin{align*}
				\AxiomC{$\tmfive_{\neg\wsym} \tow \tmfour'$}
				\UnaryInfC{$\tmtwo = \tmfive_{\neg\wsym}\tmsix \tow \tmfour'\tmsix =: \tmfour$}
				\DisplayProof
				\,;
				\end{align*}
				\item $\tm = \tmfive\tmsix \tow \tmfive\tmsix' = \tmthree$ with  $\tmsix \tow \tmsix'$; 
				hence
				\begin{align*}
				\AxiomC{$\tmsix \tow \tmsix'$}
				\UnaryInfC{$\tmtwo = \tmfive_{\neg\wsym}\tmsix \tow \tmfive_{\neg\wsym}\tmsix' =`: \tmfour$}
				\DisplayProof
				\,;
				\end{align*}		
				
				\item $\tm = (\la{\var}{\tmfive'})\tmsix \tow \tmfive'\isub{\var}{\tmsix} = \tmthree$ where $\tmfive = \la{\var}{\tmfive'}$ and $\tmsix$ is a value. 
				According to the definition of $\tm \tonotw \tmtwo$,
				\begin{align*}
				\AxiomC{$\tmfive \tobv \tmfive_{\neg\wsym}$}
				\UnaryInfC{$\la{\var}\tmfive' \tonotw \la{\var}\tmfive_{\neg\wsym}'$}
				\UnaryInfC{$\tm = (\la{\var}\tmfive')\tmsix \tonotw (\la{\var}{\tmfive_{\neg\wsym}'})\tmsix = \tmtwo$}
				\DisplayProof
				\end{align*}
				where $\tmfive_{\neg\wsym} = \la{\var}{\tmfive_{\neg\wsym}'}$; therefore
				\begin{align*}
				\AxiomC{}
				\UnaryInfC{$\tmtwo = (\la{\var}\tmfive_{\neg\wsym}')\tmsix \tow \tmfive_{\neg\wsym}'\isub{\var}{\tmsix} =: \tmfour$}
				\DisplayProof
				\,.
				\end{align*}
			\end{enumerate}
			
			\item \emph{Application right}: analogous to the previous point.
		\end{itemize}
		
		\item \emph{Diamond}: 
		The idea of the proof is that, when $(\la{\var}\tm)\val \tow \tm\isub{\var}{\val}$, the $\betav$-redexes in $\val$ (which is a value) can be duplicated in $\tm\isub{\var}{\val}$ but they are under an abstraction and $\tow$ does not reduce  under abstractions.
		
		Formally, the proof is by induction on $\tm$.
		Note that $\tm$ is not a value, because values are $\tow$-normal, since $\tow$ does not reduce under abstractions.
		Therefore, $\tm = \tm_0\tm_1$.
		The case where $\tm_0 = \la{\var}{\tm'}$ and $\tm_1$ is a value is impossible, because $\tm_0$ and $\tm_1$ would be $\tow$-normal and so from $\tmtwo \ltow \tm \tow \tmthree$ it would follow $\tmtwo = \tm'\isub{\var}{\tm_1} = \tmthree$, which contradicts the hypothesis. 
		The remaining cases for $\tm = \tm_0\tm_1$ are:
		\begin{itemize}
			\item $\tmtwo = \tm_0\tmtwo_1 \ltow \tm = \tm_0\tm_1 \tow \tmthree_0\tm_1 = \tmthree$ with $\tm_0 \tow \tmthree_0$,  and $\tm_1 \tow  \tmtwo_1$. 
			Then, $\tmtwo \tow \tmthree_0\tmtwo_1 \ltow \tmthree$.
			
			\item $\tmtwo = \tmtwo_0 \tm_1 \ltow \tm = \tm_0\tm_1 \tow \tmthree_0\tm_1 = \tmthree$ with $\tmtwo_0 \ltow \tm_0 \tow \tmthree_0$. 
			By \ih, there is $\tmfour_0$ such that $\tmtwo_0 \tow \tmfour_0 \ltow \tmthree_0$.
			Thus, $\tmtwo \tow \tmfour_0\tm_1 \ltow \tmthree$.
			
			\item $\tmtwo = \tm_0\tmtwo_1  \ltow \tm = \tm_0\tm_1 \tow \tm_0\tmthree_0 = \tmthree$ with $\tmtwo_1 \ltow \tm_1 \tow \tmthree_1$. 
			Analogous to the previous case.
			\qedhere
		\end{itemize}
	\end{enumerate}
\end{proof}

\subsection{Omitted proofs  of \refsect{leftmost} (leftmost-outermost)}

\setcounter{propositionAppendix}{\value{prop:essential-left}}
\begin{propositionAppendix}[\lo essential properties]
	\label{propappendix:essential-left}
	\NoteState{prop:essential-left}
\hfill
\begin{enumerate}
\item\label{pappendix:essential-left-completeness} \emph{\Completeness}: if $\tm \tob \tmtwo$ then there exists $\tmthree$ such that $\tm \tolo \tmthree$.
\item\label{pappendix:essential-left-determinism} \emph{Determinism}: if $\tm \tolo \tmtwo_1$ and $\tm\tolo \tmtwo_2$, then $\tmtwo_1 = \tmtwo_2$.
\item\label{pappendix:essential-left-persistence} \emph{Persistence}: if $\tm \tolo \tmtwo_1$ and $\tm \tonotlo \tmtwo_2$, then $\tmtwo_2 \tolo \tmthree$ for some $\tmthree$.
\end{enumerate}
\end{propositionAppendix}

\begin{proof}\hfill
\begin{enumerate}
	\item \emph{\Completeness}: by induction on $\tm$. Cases:
  \begin{itemize}
    \item \emph{Variable}, \ie $\tm = \var$: then $\tm \not\tob\tmtwo$, and so the statement trivially holds.
    
    \item \emph{Abstraction}, \ie $\tm = \la\var\tm' \tob \la\var\tmtwo' = \tmtwo$. It follows by the \ih
    
    \item \emph{Application}, \ie $\tm = \tmfour \tmfive$. Three sub-cases:
    \begin{itemize}
      \item \emph{$\tmfour$ is an abstraction}, \ie $\tmfour = \la\var\tmsix$: then $\tm = (\la\var\tmsix) \tmfive \tolo \tmsix\isub\var\tmsix$.
      
      \item \emph{$\tmfour$ is not an abstraction but it is not normal}, \ie $\tmfour \tob \tmfour'$ for some $\tmfour'$: then by \ih $\tmfour \tolo \tmsix$ for some $\tmsix$ and so $\tm = \tmfour \tmfive \tolo \tmsix\tmfive$.
      
      \item \emph{$\tmfour$ is neutral}, \ie $\tm$ is not normal implies $\tmfive$ not normal. Then by \ih $\tmfive \tolo \tmfive'$ for some $\tmfive'$, and so $\tm = \tmfour \tmfive \tolo \tmfour \tmfive'$.
    \end{itemize}
  \end{itemize}
  
%
%
	\item \emph{Determinism}: By induction on a derivation with conclusion $\tm \tolo \tmtwo_1$. 
	Consider its last rule.
	Cases:
	\begin{itemize}
		\item \emph{Abstraction}, \ie $\tm = \la{\var}{\tm'} \tolo \la{\var}{\tmtwo_1'} = \tmtwo_1$ because $\tm' \tolo \tmtwo_1'$.
		According to the definition of $\tolo$, the last rule of a derivation for $\tm \tolo \tmtwo_2$ can be only	 
		\begin{align*}
		\AxiomC{$\tm' \tolo \tmtwo_2'$}
		\UnaryInfC{$\tm = \la{\var}{\tm'} \tolo \la{\var}{\tmtwo_2'} = \tmtwo_2$}
		\DisplayProof
		\ .
		\end{align*}
		By \ih applied to $\tm'$, we have $\tmtwo_1' = \tmtwo_2'$ and hence $\tmtwo_1 = \la{\var}{\tmtwo_1'} = \la{\var}{\tmtwo_2'} = \tmtwo_2$.
		
		\item \emph{Application left}, \ie $\tm = \tm' \tm'' \tolo \tmtwo_1'\tm'' = \tmtwo_1$ because $\tm' \tolo \tmtwo_1'$ and $\tm' \neq \la{\var}{\tmfour}$.
		According to the definition of $\tolo$, the last rule of a derivation for $\tm \tolo \tmtwo_2$ can only be (since $\tm$ is neither an abstraction nor neutral)
		\begin{align*}
		\AxiomC{$\tm' \tolo \tmtwo_2'$}
		\UnaryInfC{$\tm = \tm'\tm'' \tolo \tmtwo_2'\tm'' = \tmtwo_2$}
		\DisplayProof
		\ .
		\end{align*}
		By \ih applied to $\tm'$, we have $\tmtwo_1' = \tmtwo_2'$ and hence $\tmtwo_1 = \tmtwo_1'\tm'' = \tmtwo_2'\tm'' = \tmtwo_2$.

		\item \emph{Application right}, \ie $\tm = \tm' \tm'' \tolo \tm'\tmtwo_1'' = \tmtwo_1$ because $\tm'' \tolo \tmtwo_1''$ and $\tm'$ is neutral.
		According to the definition of $\tolo$, the last rule of a derivation for $\tm \tolo \tmtwo_2$ can only be (since $\tm$ is normal and not an abstraction)
		\begin{align*}
		\AxiomC{$\tm'' \tolo \tmtwo_2''$}
		\UnaryInfC{$\tm = \tm'\tm'' \tolo \tm'\tmtwo_2'' = \tmtwo_2$}
		\DisplayProof
		\ .
		\end{align*}
		By \ih applied to $\tm'$, we have $\tmtwo_1' = \tmtwo_2'$ and hence $\tmtwo_1 = \tm'\tmtwo_1'' = \tm'\tmtwo_2'' = \tmtwo_2$.
		
		\item \emph{$\beta$-rule}, \ie $\tm = (\la{\var}\tm')\tm'' \tolo \tm'\isub{\var}{\tm''} = \tmtwo_1$.
		According to the definition of $\tolo$, the only possibility for the last rule of a derivation for $\tm \tolo \tmtwo_2$ is 
		\begin{align*}
		\AxiomC{}
		\UnaryInfC{$\tm = (\la{\var}{\tm'})\tm'' \tolo \tm'\isub{\var}{\tm''} = \tmtwo_2$}
		\DisplayProof
		&&\text{and hence } \tmtwo_1 = \tm'\isub{\var}{\tm''} = \tmtwo_2.
		\text{\qedhere}
		\end{align*}
	\end{itemize}

	\item \emph{Persistence}: by induction on $\tm \tolo \tmtwo_1$. Cases:
  \begin{itemize}
    \item \emph{Root}: $\tm = (\la\var \tmfour) \tmfive \tolo \tmfour\isub\var\tmfive = \tmtwo_1$. Three sub-cases:
    \begin{itemize}
      \item $\tm = (\la\var \tmfour) \tmfive \tonotlo (\la\var \tmfour') \tmfive  = \tmtwo_2$ because $\tmfour \tob \tmfour'$. Then $\tmtwo_2 = (\la\var \tmfour') \tmfive \tolo \tmfour'\isub\var\tmfive =: \tmthree$. 
      
      \item $\tm = (\la\var \tmfour) \tmfive \tonotlo (\la\var \tmfour') \tmfive  = \tmtwo_2$ because $\tmfour \tonotlo \tmfour'$. Exactly as the previous one.
      
      \item $\tm = (\la\var \tmfour) \tmfive \tonotlo (\la\var \tmfour) \tmfive ' = \tmtwo_2$ because $\tmfive \tob \tmfive'$. Then $\tmtwo_2 = (\la\var \tmfour) \tmfive' \tolo \tmfour\isub\var\tmfive' =: \tmthree$.
    \end{itemize}
    
    \item \emph{Abstraction}: $\tm = \la\var \tmfour \tolo \la\var \tmfour'= \tmtwo_1$. Then $\tm = \la\var \tmfour \tonotlo \la\var \tmfour''= \tmtwo_2$ and the statement follows from the \ih and closure of $\tolo$.
    
    \item \emph{Left of an application}: $\tm = \tmfour \tmfive \tolo \tmfour_1 \tmfive = \tmtwo_1$ with $\tmfour \tolo \tmfour_1$ and $\tmfour$ not an abstraction. Two sub-cases:    
    \begin{itemize}
      \item $\tm = \tmfour \tmfive \tonotlo \tmfour_2 \tmfive = \tmtwo_2$ because $\tmfour \tonotlo \tmfour_2$. Then by \ih there exists $\tmsix$ such that $\tmfour_2 \tolo \tmsix$. Note that $\tonotlo$ cannot create a root abstraction (because it never reduces the root redex) so that if $\tmfour$ is not an abstraction  then $\tmfour_2$ is not an abstraction and $\tmtwo_2 = \tmfour_2 \tmfive \tolo \tmsix \tmfive =:\tmthree$.
      
      \item $\tm = \tmfour \tmfive \tonotlo \tmfour \tmfive' = \tmtwo_2$ by one of the two rules for able to derive it. In both cases $\tmfive \tob \tmfive'$. Then $\tmtwo_2 = \tmfour \tmfive' \tolo \tmfour_1 \tmfive' =:\tmthree$ and $\tmtwo_1 = \tmfour_1 \tmfive \tob \tmfour_1 \tmfive' =\tmthree$.
    \end{itemize}
    
    \item \emph{Right of an application}: $\tm = \tmfour \tmfive \tolo \tmfour \tmfive_1 = \tmtwo_1$ with $\tmfive \tolo \tmfive_1$ and $\tmfour$ neutral. Then necessarily $\tm = \tmfour \tmfive \tolo \tmfour \tmfive_2 = \tmtwo_2$ with $\tmfive \tonotlo \tmfive_2$. The statement then follows by the \ih    
    \qedhere
  \end{itemize}
\end{enumerate}
\end{proof}

\setcounter{propositionAppendix}{\value{prop:macro-left}}
\begin{propositionAppendix}[\lo macro-step system]
	\label{propappendix:macro-left}
	\NoteState{prop:macro-left}
\hfill
  \begin{enumerate}
	\item\label{pappendix:macro-left-merge} \emph{Merge}: if $\tm \partonotlo \cdot \tolo \tmthree$ then $\tm \partob \tmthree$.
	\item\label{pappendix:macro-left-indexed-split} \emph{Indexed split}: if $\tm \partobind n \tmtwo$ then $\tm \partonotlo \tmtwo$, or $n>0$ and  $\tm \tolo \cdot \partobind {n-1} \tmtwo$.
	\item\label{pappendix:macro-left-split} \emph{Split}: if $\tm \partob \tmtwo$ then $\tm \tolo^* \cdot \partonotlo \tmtwo$.
	\end{enumerate}
	That is, $(\Lambda, \set{\tolo, \tonotlo})$ is a macro-step system with respect to $\partob$ and $\partonotlo$.
\end{propositionAppendix}

\begin{proof}
\hfill
  \begin{enumerate}
	\item \emph{Merge}: by induction on $\tm \partonotlo \tmtwo$. Cases:
\begin{itemize}
	\item Rule
	\[\AxiomC{$\tmfour \partob \tmfour'$}
	\AxiomC{$\tmfive \partob \tmfive''$}
	\BinaryInfC{$\tm = (\la\var \tmfour) \tmfive \partonotlo (\la\var \tmfour') \tmfive' = \tmtwo$}
	\DisplayProof\]
	Then $(\la\var \tmfour) \tmfive \partonotlo (\la\var \tmfour') \tmfive' \tolo \tmfour'\isub\var{\tmfive'} = \tmthree$. We simply have:
	\[	\AxiomC{$\tmfour \partob \tmfour'$}
	\AxiomC{$\tmfive \partob \tmfive''$}
\BinaryInfC{$\tm = (\la\var\tmfour)\tmfive \partob \tmfour'\isub\var{\tmfive'} = \tmthree$}
\DisplayProof\]
	
	\item Rule
	\[\AxiomC{$\tmfour$ not neutral}
	\AxiomC{$\tmfour \partonotlo \tmfour'$}
		\AxiomC{$\tmfive \partob \tmfive'$}
	\TrinaryInfC{$\tm = \tmfour \tmfive \partonotlo \tmfour' \tmfive ' = \tmtwo$}
	\DisplayProof\]
	Since $\tmfour$ is not neutral, it is an abstraction or it is not normal. If $\tmfour$ is an abstraction this case continues goes as the first case. Otherwise, $\tmfour$ is not normal, and by persistence (\refpropp{essential-left}{persistence}) $\tmfour'$ is not normal. 
	\Completeness (\refpropp{essential-left}{completeness}) of $\tolo$ gives $\tmfour' \tolo \tmfour''$ for a certain $\tmfour''$. The hypothesis becomes $\tm = \tmfour \tmfive \partonotlo \tmfour' \tmfive ' \tolo \tmfour'' \tmfive' = \tmthree$.
	By \ih, $\tmfour \partob \tmfour''$. Then,
	\[\AxiomC{$\tmfour \partob \tmfour''$}
	\AxiomC{$\tmfive\partob \tmfive'$}
	\BinaryInfC{$\tm = \tmfour \tmfive \partob \tmfour'' \tmfive' = \tmthree$}
	\DisplayProof\]

	\item Rule
	\[\AxiomC{$\tmfour \partonotlo \tmfour'$}
	\UnaryInfC{$\tm = \la\var\tmfour \partonotlo \la\var\tmfour' = \tmtwo$}
	\DisplayProof \]
	Then $\la\var\tmfour \partonotlo \la\var\tmfour' \tolo \la\var\tmfour'' = \tmthree$ with $\tmfour' \tolo \tmfour''$. By \ih, $\tmfour \partob \tmfour''$ and 
	\[\AxiomC{$\tmfour \partob \tmfour'$}
\UnaryInfC{$\tm = \la \var\tmfour \partob \la \var \tmfour'' = \tmthree$}
\DisplayProof\]

	\item Rule 
	\[\AxiomC{$\tmfour$ neutral }
	\AxiomC{$\tmfive \partonotlo \tmfive'$}
	\BinaryInfC{$\tm = \tmfour \tmfive \partonotlo \tmfour \tmfive' = \tmtwo$}
	\DisplayProof  \]
	 Then the hypothesis is $\tmfour \tmfive \partonotlo \tmfour \tmfive' \tolo \tmfour \tmfive'' = \tmthree$ with $\tmfive' \tolo \tmfive''$. By \ih, $\tmfive \partob \tmfive''$, and since $\partob$ is reflexive,
		\[\AxiomC{$\tmfour \partob \tmfour$}
		\AxiomC{$\tmfive\partob \tmfive''$}
		\BinaryInfC{$\tm = \tmfour \tmfive \partob \tmfour \tmfive'' = \tmthree$}
		\DisplayProof\]
\end{itemize}
	\item \emph{Indexed split}: By induction on the definition of $\tm \partobind n \tmtwo$. We use freely the fact that if $\tm \partobind n \tmtwo$ then $\tm \partob \tmtwo$. Cases:
\begin{itemize}
	\item \emph{Variable}: $\tm = \var \partobind 0 \var = \tmtwo$. Then $\tm = \var \partonotlo \var = \tmtwo$.

	\item \emph{Abstraction}: $\tm = \la \var\tm' \partobind n  \la \var \tmtwo' = \tmtwo$ because $\tm' \partobind n \tmtwo'$. It follows from the \ih.

	\item \emph{Application}:
	\begin{align*}
	\label{eq:app}
	\AxiomC{$\tmfour \partobind {n_1} \tmfour'$}
	\AxiomC{$\tmfive \partobind {n_2} \tmfive'$}
	\BinaryInfC{$\tm = \tmfour \tmfive \partobindlong {n_1 + n_2} \tmfour' \tmfive' = \tmtwo$}
	\DisplayProof
	\end{align*}
	with $n = n_1 + n_2$. There are only two cases:
	\begin{itemize}
		\item either $\tmfour \tmfive \partonotlo \tmfour' \tmfive'$, and then the claim holds;
		\item or $\tmfour \tmfive \not\partonotlo \tmfour' \tmfive'$, then the following conditions hold (otherwise $\tmfour \tmfive \partonotlo \tmfour' \tmfive'$):
		\begin{enumerate}
			\item $\tmfour$ is not an abstraction;
			\item if $\tmfour$ is neutral then $\tmfive \not\partonotlo \tmfive'$;
			\item if $\tmfour$ is not neutral then $\tmfour \not\partonotlo \tmfour'$; 
		\end{enumerate} 
		So, if $\tmfour$ is neutral, then by \ih applied to $\tmfive \partobind{n_2} \tmfive'$, $n_2 > 0$ and there is $\tmfive''$ such that $\tmfive \tolo \tmfive'' \partobind{n_2-1} \tmfour'$; thus, $\tm = \tmfour \tmfive \tolo \tmfour\tmfive''$ and 
		\[
		\AxiomC{$\tmfour \partobind {n_1} \tmfour'$}	
		\AxiomC{$\tmfive'' \partobindlong {n_2-1} \tmfive'$}		
		\BinaryInfC{$\tmfour \tmfive'' \partobindlong {n_1 + n_2-1} \tmfour' \tmfive' = \tmtwo$}
		\DisplayProof
		\ .
		\]
		
		Otherwise $\tmfour$ is not neutral and hence, by \ih applied to $\tmfour \partobind{n_1} \tmfour'$,
		$n_1>0$ and there exists $\tmfour''$ such that $\tmfour \tolo \tmfour'' \partobind {n_1-1} \tmfour'$;
		thus, $\tm = \tmfour \tmfive \tolo \tmfour'' \tmfive$ and 
		\[\AxiomC{$\tmfour'' \partobindlong {n_1-1} \tmfour'$}			
		\AxiomC{$\tmfive \partobind {n_2} \tmfive'$}	
		\BinaryInfC{$\tmfour'' \tmfive \partobindlong {n_1-1 + n_2} \tmfour' \tmfive' = \tmtwo$}
		\DisplayProof
		\ .
		\]
	\end{itemize}

	\item \emph{$\beta$ step}:
		\[\AxiomC{$\tmthree \partobind {n_1} \tmthree'$}
		\AxiomC{$\tmfour \partobind {n_2} \tmfour''$}		
		\BinaryInfC{$\tm = (\la\var \tmthree)\tmfour \partobindlong {n_1 + \sizep{\tmtwo}\var \cdot n_2 +1} \tmthree'\isub \var {\tmfour'} = \tmtwo$}
		\DisplayProof\]
		With $n = n_1 + \sizep\tmtwo\var \cdot n_2 +1 > 0$. We have $\tm = (\la\var \tmthree)\tmfour \tolo \tmthree\isub\var\tmfour$ and by substitutivity of $\partobind n$ (\reflemma{partobind-subs}) $\tmthree\isub\var\tmfour \partobindlong{n_1 + \sizep\tmtwo\var \cdot n_2} \tmthree' \isub\var{\tmfour'} = \tmtwo$.
\end{itemize}

	\item \emph{Split}: exactly as in the head case (\refpropp{macro-head}{split}), using the Indexed Split property for \lo (\refpoint{macro-left-indexed-split} above).
	\qedhere
	\end{enumerate} 
\end{proof}

\subsection{Omitted proofs and lemmas of \refsect{stratified-cbn} (least-level)}

\setcounter{propositionAppendix}{\value{prop:deg}}
\begin{propositionAppendix}[Least level properties]
	\label{propappendix:deg}
	\NoteState{prop:deg}
\begin{enumerate}
	\item\label{pappendix:deg-finite}\emph{Computational meaning of $\llsym$}: $\Deg{\tm} = \inf\{ k \in \nat \mid \tm \tostratind{k} \tmthree \text{ for some term } \tmthree\}$.
	\item\label{pappendix:deg-invariance-leq} \emph{Monotonicity}: if $\tm \tob \tmtwo$ then $\Deg{\tmtwo} \geq \Deg{\tm}$. 
	\item\label{pappendix:deg-invariance-equal} \emph{Invariance by $\tonotstrat$}: if $\tm \tonotstrat \tmtwo$ then $\Deg{\tmtwo} = \Deg{\tm}$. 
\end{enumerate}
\end{propositionAppendix}

\begin{proof}\hfill
	\begin{enumerate}
		\item 
		By induction on $\tm$.
		For any term $\tmfour$, we set $\inf_\tmfour = \inf\{ k \in \nat \mid \tmfour \tostratind{k} \tmthree \text{ for some term } \tmthree \}$.
		Cases:
		
		\begin{itemize}
			\item \emph{Variable}, \ie $\tm$ is a variable. Then, $\Deg{\tm} = \infty$ and $\tm$ is $\tob$-normal.
			
			\item \emph{Abstraction}, \ie $\tm = \la{\var}{\tmtwo}$.
			Then, $\Deg{\tm} = \Deg{\tmtwo}$ and, by \ih, $\Deg{\tmtwo} = \inf_\tmtwo$;
			now,
			\begin{align*}
			\AxiomC{$\tmtwo \tostratind{k} \tmthree$}
			\UnaryInfC{$\tm = \la{\var}\tmtwo \tostratind{k} \la{\var}\tmthree$}
			\DisplayProof
			\end{align*}
			and there is no other rule for $\tostratind{n}$ whose conclusion is of the form $\la{\var}{\tmtwo} \tostratind{n} \tmfive$;
			therefore, $\Deg{\tm} = \Deg{\tmtwo} = \inf_\tmtwo = \inf_\tm$.
			
			\item \emph{Application}, \ie $\tm = \tm'\tm''$.
			There are two sub-cases:
			\begin{itemize}
				\item $\tm' = \la{\var}\tmtwo'$, then $\Deg{\tm} = 0$ and
				\begin{align*}
				\AxiomC{}
				\UnaryInfC{$\tm = (\la{\var}\tmtwo')\tm'' \tostratind{0} \tmtwo'\isub{\var}{\tm''} $}
				\DisplayProof
				\end{align*}
				thus $\inf_\tm = 0 = \Deg{\tm}$.
				\item $\tm'$ is not an abstraction, then $\Deg{\tm} = \min \{\Deg{\tm'}, \Deg{\tm''}+1\}$.
				By \ih, $\Deg{\tm'} = \inf_{\tm'}$ and $\Deg{\tm''} = \inf_{\tm''}$; 
				now,
				\begin{align*}
				\AxiomC{$\tm' \tostratind{k} \tmtwo'$}
				\UnaryInfC{$\tm = \tm'\tm'' \tostratind{k} \tmtwo'\tm''$}
				\DisplayProof
				&&\text{and}&&
				\AxiomC{$\tm'' \tostratind{k} \tmtwo''$}
				\UnaryInfC{$\tm = \tm'\tm'' \tostratind{k+1} \tm'\tmtwo''$}
				\DisplayProof
				\end{align*}
				and there is no other rule for $\tostratind{n}$ whose conclusion is of the form $\tm'\tm'' \tostratind{n} \tmfive$ (as $\tm'$ is not an abstraction);
				hence, $\Deg{\tm} = \min\{\inf_{\tm'}, \inf_{\tm''}+1\} = \inf_\tm$.
			\end{itemize}
		\end{itemize}
		
		\item \emph{Monotonicity}: by induction on the definition of $\tm \tob \tmtwo$. 
		Cases:
		\begin{itemize}
			\item \emph{Abstraction}, \ie $\tm = \la{\var}{\tm'} \tob \la{\var}{\tmtwo'} = \tmtwo$ because $\tm' \tob \tmtwo'$.
			Then, $\Deg{\tm} = \Deg{\tm'} \leq \Deg{\tmtwo'} = \Deg{\tmtwo}$ by \ih
			
			\item \emph{Application left}, \ie $\tm = \tm' \tm''  \tob \tmtwo'\tm'' = \tmtwo$ because $\tm' \tob \tmtwo'$.
			By \ih, $\Deg{\tm'} \leq \Deg{\tmtwo'}$.
			Hence, $\Deg{\tm} = \min \{\Deg{\tm'}, \Deg{\tm''} +1\} \leq \min \{ \Deg{\tmtwo'}, \Deg{\tm''} + 1\} = \Deg{\tmtwo}$.
			
			\item \emph{Application right}, \ie $\tm = \tm' \tm''  \tob \tm'\tmtwo'' = \tmtwo$ because $\tm'' \tob \tmtwo''$.
			Analogous to the previous case.
			
			\item \emph{$\beta$-redex}, \ie $\tm = (\la{\var}\tm')\tm'' \tob \tm'\isub{\var}{\tm''} = \tmtwo$. 
			Then, $\Deg{\tm} = 0 \leq \Deg{\tmtwo}$.  
		\end{itemize}
		
		\item \emph{Invariance by $\tonotstrat$}: By hypothesis, $\tm \tostratind{n} \tmtwo$ for some $n > \Deg{\tm}$.
		We proceed by induction on the definition of $\tm \tostratind{n} \tmtwo$.
		Cases:
		\begin{itemize}			
			\item \emph{Abstraction}: $\tm = \la \var\tm' \tostratind{n}  \la \var \tmtwo' = \tmtwo$ because $\tm' \tostratind{n} \tmtwo'$.
			As $n > \Deg{\tm} = \Deg{\tm'}$, then $\Deg{\tmtwo} = \Deg{\tmtwo'} = \Deg{\tm'} = \Deg{\tm}$ by \ih
			
			\item \emph{Application left}, \ie $\tm = \tm' \tm'' \tostratind{n} \tmtwo' \tm'' = \tmtwo$ because $\tm' \tostratind{n} \tmtwo'$.
			According to the definition of $\Deg{\tm}$, there are the following sub-cases:
			\begin{enumerate}
				\item $\Deg{\tm} = \Deg{\tm'} \leq \Deg{\tm''} + 1$ and $\tm'$ is not an abstraction; 
				by \ih applied to $\tm'$ (since $\Deg{\tm'} = \Deg{\tm} < n$), we have $\Deg{\tmtwo'} = \Deg{\tm'}$, and so  $\Deg{\tmtwo} = \min \{\Deg{\tmtwo'}, \Deg{\tm''} + 1\} = \min\{\Deg{\tm'}, \Deg{\tm''} + 1\} = \Deg{\tm}$.
				
				\item $\Deg{\tm} = \Deg{\tm''} + 1 \leq \Deg{\tm'}$ and $\tm'$ is not an abstraction; 
				by \refpropp{deg}{invariance-leq} (since $\tostratind{n} \subseteq \tob$),
				$\Deg{\tm'} \leq \Deg{\tmtwo'}$ and therefore 
				$\Deg{\tmtwo} = \min \{\Deg{\tmtwo'}, \Deg{\tm''} + 1\} = \Deg{\tm''}+1 = \min\{\Deg{\tm'}, \Deg{\tm''} + 1\} = \Deg{\tm}$.
				
				\item $\Deg{\tm} = 0$ and $\tm'= \la{\var}\tmthree'$. 
				According to the definition of $\tm \tostratind{n} \tmtwo$, since $n > 0$, we have
				\begin{align*}
				\AxiomC{$\tmthree' \tostratind{n} \tmtwo'$}
				\UnaryInfC{$\la{\var}\tmthree' \tostratind{n} \la{\var}{\tmtwo'}$}
				\UnaryInfC{$\tm = (\la{\var}\tmthree')\tm'' \tostratind{n} (\la{\var}{\tmtwo'})\tm'' = \tmtwo$}
				\DisplayProof
				&&
				\text{or}
				&&
				\AxiomC{$\tm'' \tostratind{n-1} \tmtwo''$}
				\UnaryInfC{$\tm = (\la{\var}\tmthree')\tm'' \tostratind{n} (\la{\var}{\tmthree'})\tmtwo'' = \tmtwo$}
				\DisplayProof
				\end{align*}
				therefore $\Deg{\tmtwo} = 0 = \Deg{\tm}$.
			\end{enumerate}
			
			\item \emph{Application right}, \ie $\tm = \tm' \tm'' \tostratind{n} \tm' \tmtwo'' = \tmtwo$ because $\tm'' \tostratind{n-1} \tmtwo''$.
			Analogous to the previous case.
			
			\item \emph{$\beta$-step}, \ie $\tm = (\la{\var}\tm')\tm'' \tostratind{0} \tm' \isub{\var}{\tm''} = \tmtwo$
			where $0 = n > \Deg{\tm} = 0$, which is impossible.
			\qedhere
		\end{itemize}
	\end{enumerate}
\end{proof}

\begin{lemma}
	\label{l:not-abs}
	Let $\tm \tostrat \tmtwo$ with $\Deg{\tm} > 0$.
	If $\tm$ is not an abstraction, then $\tmtwo$ is not an abstraction. 
\end{lemma}

\begin{proof}
	By hypothesis, $\tm = \tm'\tm''$ and $\tm'$ is not an abstraction (otherwise $\Deg{\tm} = 0$).
	Therefore, according to the definition of $\tostrat$, there are only two possibilities:
	\begin{enumerate}
		\item either $\tm = \tm' \tm'' \tostratind{k} \tmtwo'\tm'' = \tmtwo$ because $\tm' \tostratind{k} \tmtwo'$ and $k = \Deg{\tm}$;
		\item or $\tm = \tm'\tm'' \tostratind{k} \tm' \tmtwo'' = \tmtwo$ because $\tm'' \tostratind{k - 1} \tmtwo''$ and $k = \Deg{\tm}$.
	\end{enumerate}
	In both cases, $\tmtwo$ is not an abstraction.
\end{proof}

\begin{lemma}[Substitutivity by level]
	\label{l:substitutivity-strat}
	If $\tm \tostratind{k} \tmtwo$ then $\tm\isub{\var}{\tmthree} \tostratind{k} \tmtwo\isub{\var}{\tmthree}$.
%
%
\end{lemma}

\begin{proof}
		By induction on the definition of $\tm \tostratind{k} \tmtwo$.
		Cases:
		\begin{itemize}
			\item \emph{Abstraction}: $\tm = \la{\vartwo}{\tm'} \tostratind{k} \la{\vartwo}{\tmtwo'} = \tmtwo$ with $\tm' \tostratind{k} \tmtwo'$. 
			We can suppose without loss of generality that $\vartwo \notin \fv{\tmthree} \cup \{\var\}$.
			By \ih, $\tm'\isub{\var}{\tmthree} \tostratind{k} \tmtwo'\isub{\var}{\tmthree}$ and hence $\tm\isub{\var}{\tmthree} = \la{\vartwo}\tm'\isub{\var}{\tmthree} \tostratind{k} \la{\vartwo}\tmtwo'\isub{\var}{\tmthree} = \tmtwo\isub{\var}{\tmthree}$.
			
			\item \emph{Application left}: $\tm = \tm'\tm'' \tostratind{k} \tmtwo'\tm'' = \tmtwo$ with $\tm' \tostratind{k} \tmtwo'$. 
			By \ih, $\tm'\isub{\var}{\tmthree} \tostratind{k} \tmtwo'\isub{\var}{\tmthree}$ and hence $\tm\isub{\var}{\tmthree} = \tm'\isub{\var}{\tmthree} \tm''\isub{\var}{\tmthree} \tostratind{k} \tmtwo'\isub{\var}{\tmthree} \tm''\isub{\var}{\tmthree} = \tmtwo\isub{\var}{\tmthree}$.
			
			\item \emph{Application right}: $\tm = \tm'\tm'' \tostratind{k} \tm'\tmtwo'' = \tmtwo$ with $k > 0$ and $\tm'' \tostratind{k} \tmtwo''$. 
			Analogous to the previous case.

			\item \emph{$\beta$-redex}: $\tm = (\la{\vartwo}\tm')\tm'' \tostratind{0} \tm'\isub{\vartwo}{\tm''} = \tmtwo$.
			We can suppose without loss of generality that $\vartwo \notin \fv{\tmthree} \cup \{\var\}$.
			Then, $\tm\isub{\var}{\tmthree} = (\la{\vartwo}\tm'\isub{\var}{\tmthree})\tm''\isub{\var}{\tmthree} \tostratind{0} \tm'\isub{\var}{\tmthree}\isub{\vartwo}{\tm''\isub{\var}{\tmthree}} = \tm'\isub{\vartwo}{\tm''}\isub{\var}{\tmthree} = \tmtwo\isub{\var}{\tmthree}$.
			\qedhere
		\end{itemize}
	
	
\end{proof}

\setcounter{propositionAppendix}{\value{prop:essential-strat}}
\begin{propositionAppendix}[\ll essential properties]
	\label{propappendix:essential-strat}
	\NoteState{prop:essential-strat}\hfill
	\begin{enumerate}
		\item\label{pappendix:essential-strat-completeness}
		\emph{\Completeness}: if $\tm \tob \tmtwo$ then $\tm \toll \tmthree$ for some $\tmthree$.
		\item\label{pappendix:essential-strat-persistence} \emph{Persistence}: if $\tm \toll \tmtwo_1$ and $\tm \tonotll \tmtwo_2$ then $\tmtwo_2 \toll \tmthree$ for some $\tmthree$.
		\item\label{pappendix:essential-strat-diamond} \emph{Diamond}: if $\tmtwo\ltoll \cdot \toll\tmthree$ with $\tmtwo \neq \tmthree$ then $\tmtwo \toll \cdot \ltoll \tmthree$.
	\end{enumerate}
\end{propositionAppendix}

\begin{proof}\hfill
	\begin{enumerate}
		\item \emph{\Completeness}: Since $\tm$ is not $\tob$-normal, then $\infty \neq \Deg{\tm} =  \min\{k \in \nat \mid \tm \tostratind{k} \tmfour \text{ for some } \tmfour\}$ by \refpropp{deg}{finite}.
		Therefore, $\tm \tostrat \tmthree$  for some $\tmthree$.
		
		\item \emph{Persistence}: Since $\tm$ is not $\tob$-normal, $\infty \neq \Deg{\tm} =  \min\{k \in \nat \mid \tm \tostratind{k} \tmfour \text{ for some } \tmfour\}$ by \refpropp{deg}{finite}.
		By least level invariance (\refpropp{deg}{invariance-equal}), $\Deg{\tm} = \Deg{\tmtwo_2}$ and hence, according to \refpropp{deg}{finite} again, we have $\infty \neq \Deg{\tmtwo_2} = \min\{k \in \nat \mid \tmtwo_2 \tostratind{k} \tmthree \text{ for some } \tmthree\}$.
		Therefore, there exists $\tmthree$ such that $\tmtwo_2 \tostrat \tmthree$.

				\item \emph{Diamond}: 
		The idea of the proof is that, when $(\la{\var}\tm)\tmtwo \tostrat \tm\isub{\var}{\tmtwo}$, the $\beta$-redexes in $\tmtwo$ can be duplicated in $\tm\isub{\var}{\tmtwo}$ but they are not at least level and $\tostrat$ does not reduce outside the least level.
		
		Formally, according to the definition of $\tostrat$, we have to prove that $\tmtwo \ltostratind{k} \tm \tostratind{k} \tmthree$ with $k = \Deg{\tm}$, then $\tmtwo \tostratind{m} \tmfour \ltostratind{n} \tmthree$ for some $\tmfour$, where $m = \Deg{\tmtwo}$ and $n = \Deg{\tmthree}$.
		To get the right \ih, we prove also that $\Deg{\tmtwo} = k = \Deg{\tmthree}$ (and so $\tmtwo \tostratind{k} \tmfour \ltostratind{k} \tmthree$). 
		
		Clearly, $\tm$ is not a variable, otherwise it would be $\tob$-normal.
		
		If $\tm = \la{\var}{\tm'}$, then $\tmtwo = \la{\var}{\tmtwo'}$ and $\tmthree = \la{\var}{\tmthree'}$ with $\tmtwo' \ltostratind{k} \tm' \tostratind{k} \tmthree'$ and $\tmtwo' \neq \tmthree'$ and $\Deg{\tm} = \Deg{\tm'}$. 
		By \ih, $\tmtwo' \tostratind{m} \tmfour \ltostratind{n} \tmthree'$ for some term $\tmfour$, with $\Deg{\tmtwo} = \Deg{\tmtwo'} = k = \Deg{\tmthree'} = \Deg{\tmthree}$, hence $\tmtwo = \la{\var}\tmtwo' \tostratind{k} \la{\var}{\tmfour} \ltostratind{k} \la{\var}{\tmthree'} = \tmthree$.
		
		Finally, consider $\tm = \tm_0\tm_1$.
		The case where $\tm_0 = \la{\var}{\tm_2}$ and $\tmtwo = \tm_2\isub{\var}{\tm_1} \ltostratind{k} \tm \tostratind{k} (\la{\var}\tm_2)\tmthree_1 = \tmthree$ with $\tm_1 \tostratind{k-1} \tmthree_1$ is impossible, because $k = \Deg{\tm} = 0$. 
		The remaining cases for $\tm = \tm_0\tm_1$ are:
		\begin{itemize}
			\item $\tmtwo = \tm_0\tmtwo_1 \ltostratind{k} \tm = \tm_0\tm_1 \tostratind{k} \tmthree_0\tm_1 = \tmthree$ with $\tm_0 \tostratind{k} \tmthree_0$  and $\tm_1 \tostratind{k} \tmtwo_1$ and $ \Deg{\tm_0} = \Deg{\tm} = \Deg{\tm_1} + 1 = k > 0$. 
			Then, $\tm_0$ is not an abstraction (otherwise $\Deg{\tm} = 0$) and, by \reflemma{not-abs}, $\tmthree_0$ is not an abstraction. 
			By \refpropp{deg}{invariance-leq}, $\Deg{\tmtwo_1} \geq \Deg{\tm_1}$ and $\Deg{\tmthree_0} \geq \Deg{\tm_0}$.
			Hence, $\Deg{\tmtwo} = \min\{ \Deg{\tm_0}, \Deg{\tmtwo_1}+1 \} = \Deg{\tm_0} = \Deg{\tm_1}+1 = \min\{\Deg{\tmthree_0}, \Deg{\tm_1}+1\} = \Deg{\tmthree}$ and so $\tmtwo \tostratind{k} \tmthree_0\tmtwo_1 \ltostratind{k} \tmthree$.
			
			\item $\tmtwo = \tmtwo_0 \tm_1 \ltostratind{k} \tm = \tm_0\tm_1 \tostratind{k} \tmthree_0\tm_1 = \tmthree$ with $\tmtwo_0 \ltostratind{k} \tm_0 \tostratind{k} \tmthree_0$ and $\tmtwo_0 \neq \tmthree_0$ and $k = \Deg{\tm} = \Deg{\tm_0} \leq \Deg{\tm_1}+1$. 
			By \ih, there is $\tmfour_0$ such that $\tmtwo_0 \tostratind{k} \tmfour_0 \ltostratind{k} \tmthree_0$ where $\Deg{\tmtwo_0} = k = \Deg{\tmthree_0}$.
			Thus, $\tmtwo \tostratind{k} \tmfour_0\tm_1 \ltostratind{k} \tmthree$.
			If $\tmtwo_0$ or $\tmthree_0$ is an abstraction, then $\Deg{\tmtwo} = 0$ or $\Deg{\tmthree} = 0$ and hence (by \refpropp{deg}{invariance-leq}) $0 = \Deg{\tmtwo} \geq \Deg{\tm}$ or $0 = \Deg{\tmtwo} \geq \Deg{\tm}$, so in both cases $k = \Deg{\tm} = 0 = \Deg{\tmtwo} = \Deg{\tmthree}$. 
			Otherwise, $\Deg{\tmtwo} = \min\{\Deg{\tmtwo_0}, \Deg{\tm_1}+1 \} = k = \min\{\Deg{\tmthree_0}, \Deg{\tm_1}+1 \} = \Deg{\tmthree}$.
			
			\item $\tmtwo = \tm_0\tmtwo_1  \ltostratind{k} \tm = \tm_0\tm_1 \tostratind{k} \tm_0\tmthree_0 = \tmthree$ with $\tmtwo_1 \ltostratind{k-1} \tm_1 \tostratind{k-1} \tmthree_1$ and $k = \Deg{\tm} = \Deg{\tm_1}+1$. 
			Analogous to the previous case. 
			
			\item $\tmtwo = \tm_2\isub{\var}{\tm_1} \ltostratind{k} \tm = (\la{\var}\tm_2)\tm_1 \tostratind{k} (\la{\var}{\tmtwo_2})\tm_1  = \tmthree$ where $\tm_0 = \la{\var}{\tm_2}$ and $\tm_2 \tostratind{k} \tmtwo_2$ and $k = \Deg{\tm_2} = \Deg{\tm} = 0 = \Deg{\tmthree}$.
			Thus, $\tm_2 \tostratind{0} \tmtwo_2$ and $\tm \tostratind{0} \tm_2\isub{\var}{\tm_1}$.
			By substitutivity by level (\reflemma{substitutivity-strat}), $\tm_2\isub{\var}{\tm_1} \tostratind{0} \tmtwo_2\isub{\var}{\tm_1}$ and so $\Deg{\tm_2\isub{\var}{\tm_1}} =0$ by \refpropp{deg}{finite}.
			Therefore $\tm_2\isub{\var}{\tm_1} \tostratind{0} \tmtwo_2\isub{\var}{\tm_1} \ltostratind{0} \tmthree$. 
			\qedhere
		\end{itemize}
	\end{enumerate}
\end{proof}

\begin{lemma}[Merge by level]
	\label{l:merge-by-level}
	If $\tm \partostratind{n}\tmtwo \tostratind{m} \tmthree$ with $n > m$, then $\tm \partob \tmthree$.
\end{lemma}

\begin{proof}
	By induction on the definition of $\tm \partostratind{n} \tmtwo$.
	Cases:
	\begin{itemize}
		\item \emph{Variable}: $\tm = \var \partostratind{\infty} \var = \tmtwo$. 
		Then, there is no $\tmthree$ such that $\tmtwo \tostratind{m} \tmthree$ for any $m \in \nat$.
		
		\item \emph{Abstraction}: $\tm = \la \var\tm' \partostratind{n}  \la \var \tmtwo' = \tmtwo$ because $\tm' \partostratind{n} \tmtwo'$. 
		According to the definition of $\tmtwo \tostratind{m} \tmthree$, by necessity $\tmthree = \la{\var}{\tmthree'}$ with $\tmtwo' \tostratind{m} \tmthree'$.
		By \ih, $\tm' \parto \tmthree'$, thus 
		\begin{align*}
		\AxiomC{$\tm' \parto \tmthree'$}
		\UnaryInfC{$\tm = \la{\var}{\tm'} \parto \la{\var}{\tmthree'} = \tmthree$}
		\DisplayProof
		\,.
		\end{align*}
		
		\item \emph{Application}:
		\begin{align*}
		\AxiomC{$\tm_0 \partostratind{n_0} \tmtwo_0$}	
		\AxiomC{$\tm_1 \partostratind{n_1} \tmtwo_1$}
		\BinaryInfC{$\tm = \tm_0 \tm_1 \partostratind{n} \tmtwo_0 \tmtwo_1 = \tmtwo$}
		\DisplayProof 
		\end{align*}
		where $n = \min\{n_0,n_1+1\}$.
		According to the definition of $\tmtwo \tostratind{m} \tmthree$, there are the following sub-cases:
		\begin{enumerate}
			\item $\tmtwo = \tmtwo_0\tmtwo_1 \tostratind{m} \tmthree_0\tmtwo_1 = \tmthree$ with $\tmtwo_0 \tostratind{m} \tmthree_0$; 
			since $m < n \leq n_0$, 
			by \ih applied to $\tm_0 \partostratind{n_0} \tmtwo_0 \tostratind{m} \tmthree_0$, we have $\tm_0 \partob \tmthree_0$, and so (as $\partostratind{n_1} \,\subseteq\, \partob$)
			\begin{align*}
			\AxiomC{$\tm_0 \partob \tmthree_0$}
			\AxiomC{$\tm_1 \partob \tmtwo_1$}
			\BinaryInfC{$\tm = \tm_0\tm_1 \partob \tmthree_0\tmtwo_1 = \tmthree$}
			\DisplayProof
			\,;
			\end{align*}
			\item $\tmtwo = \tmtwo_0\tmtwo_1 \tostratind{m} \tmthree_0\tmtwo_1 = \tmthree$ with $\tmtwo_1 \tostratind{m-1} \tmthree_1$; 
			since $m-1 < n-1 \leq n_1$,
			by \ih applied to $\tm_1 \partostratind{n_1} \tmtwo_1 \tostratind{m-1} \tmthree_1$, we have $\tm_1 \parto \tmthree_1$, and so (as $ \partostratind{n_0} \,\subseteq\, \parto$)
			\begin{align*}
			\AxiomC{$\tm_0 \parto \tmtwo_0$}
			\AxiomC{$\tm_1 \parto \tmthree_1$}
			\BinaryInfC{$\tm = \tm_0\tm_1 \parto \tmtwo_0\tmthree_1 = \tmthree$}
			\DisplayProof
			\,;
			\end{align*}		
			\item $\tmtwo = (\la{\var}{\tmtwo_0'})\tmtwo_1 \tostratind{0} \tmtwo_0\isub{\var}{\tmtwo_1} = \tmthree$ with $\tmtwo_0 = \la{\var}\tmtwo_0'$ and $m = 0$; 
			as $n > 0$ then, according to the definition of $\tm \partostratind{n} \tmtwo$,
			\begin{align*}
			\AxiomC{$\tm_0 \partostratind{n_0} \tmtwo_0$}
			\UnaryInfC{$\la{\var}\tm_0 \partostratind{n_0} \la{\var}\tmtwo_0$}
			\AxiomC{$\tm_1 \partostratind{n_1} \tmtwo_1$}
			\BinaryInfC{$\tm = (\la{\var}\tm_0)\tm_1 \partostratind{n} (\la{\var}{\tmtwo_0})\tmtwo_1 = \tmtwo$}
			\DisplayProof
			\end{align*}
			where $n = \min\{n_0, n_1+1\}$; therefore (as $\partostratind{k} \,\subseteq\, \parto$)
			\[
			\AxiomC{$\tm_0 \parto \tmtwo_0$}
			\AxiomC{$\tm_1 \parto \tmtwo_1$}
			\BinaryInfC{$\tm = (\la{\var}{\tm_0})\tm_1 \parto \tmtwo_0\isub{\var}{\tmtwo_1} = \tmthree$}
			\DisplayProof
			\,.\qedhere
			\]
		\end{enumerate}
	\end{itemize}
\end{proof}

\setcounter{propositionAppendix}{\value{prop:macro-strat}}
\begin{propositionAppendix}[\ll macro-step system]
	\hfill\label{propappendix:macro-strat}
	\NoteState{prop:macro-strat}
	\begin{enumerate}
		\item\label{pappendix:macro-strat-merge} \emph{Merge}: if $\tm \partonotstrat \tmtwo \tostrat \tmthree$, then $\tm \partob \tmthree$.
		\item\label{pappendix:macro-strat-indexed-split}\emph{Indexed split}:	if $\tm \partobind n \tmtwo$ then $\tm \partonotstrat\tmtwo$, or $n>0$ and $\tm \tostrat \cdot \partobindlong {n-1} \tmtwo$.
		\item\label{pappendix:macro-strat-split} \emph{Split}: if $\tm \partob \tmtwo$ then $\tm \tost^* \cdot \partonotstrat\tmtwo$.
	\end{enumerate}
	That is, $(\Lambda, \set{\toll, \tonotll})$ is a macro-step system with respect to $\partob$ and $\partonotll$.
\end{propositionAppendix}

\begin{proof}\hfill
	\begin{enumerate}
		\item \emph{Merge}: As $\tm \partonotstrat \tmtwo \tostrat \tmthree$, then $\tm \partostratind{n} \tmtwo \tostratind{m} \tmthree$ for some $n \in \nat \cup \{\infty\}$ and $m \in \nat$ such that $n > \Deg{\tm}$ and  $m = \Deg{\tmtwo}$.
		Since $\partonotstrat \subseteq \tonotstrat^*$ and $\tonotstrat$ cannot change the least level (\refpropp{deg}{invariance-equal}),  $\Deg{\tm} = \Deg{\tmtwo}$
		and so $n > m$.
		By Merge by Level (\reflemma{merge-by-level}), $\tm \partob \tmthree$.
		
		\item \emph{Indexed split}: By induction on the definition of $\tm \partobind n \tmtwo$. 
		We freely use the fact that if $\tm \partobind n \tmtwo$ then $\tm \partob \tmtwo$. Cases:
		\begin{itemize}
			\item \emph{Variable}: $\tm = \var \partobind 0 \var = \tmtwo$. 
			Then, $\tm = \var \partonotst\var = \tmtwo$  since $\var \partostratind{\infty} \var$.
			
			\item \emph{Abstraction}: $\tm = \la \var\tm' \partobind n  \la \var \tmtwo' = \tmtwo$ because $\tm' \partobind n \tmtwo'$. It follows by the \ih.
			
			\item \emph{Application}:
			\begin{align}
			\label{eq:app-parallel}
			\AxiomC{$\tmr \partobind {n_1} \tmr'$}
			\AxiomC{$\tmfive \partobind {n_2} \tmfive'$}
			\BinaryInfC{$\tm = \tmfour \tmfive \partobindlong {n_1 + n_2} \tmfour' \tmfive' = \tmtwo$}
			\DisplayProof
			\end{align}
			with $n = n_1 + n_2$.
			There are only two cases:
			\begin{itemize}
				\item either $\tmfour \tmfive \partonotst \tmfour' \tmfive'$, and then the claim holds;
				\item or $\tmfour\tmfive \not\partonotstrat \tmfour'\tmfive'$ and hence any derivation with conclusion $\tmr \tmfive \partostratind{d} \tmr' \tmfive'$ is such that $d =\deg{\tmr \tmfive} \in \nat$.
				Let us rewrite the derivation \eqref{eq:app-parallel} replacing $\partobind{n}$ with $\partostratind{k}$: we have\footnote{This is possible because the inference rules for $\partobind{n}$ and $\partostratind{k}$ are the same except for the way they manage their own indexes $n$ and $k$.}
				\[\AxiomC{$\tmr \partoat {d_{\tmr}} \tmr'$}
				\AxiomC{$\tmfive \partoat{d_{\tmfive}} \tmfive'$}
				\BinaryInfC{$\tm = \tmr \tmfive \partoat d \tmr' \tmfive' = \tmtwo$}
				\DisplayProof\]
				where $d = \min\{d_\tmfour, d_\tmfive +1 \}$.
				Thus, there are two sub-cases:
				\begin{enumerate}
					\item $d = d_\tmfour \leq d_\tmfive+1$ and then $d=\deg  {\tmfour \tmfive}\leq \deg \tmfour \leq d_{\tmfour}=d $ (the first inequality holds by definition of $\Deg{\tmfour\tmfive}$), hence $\deg \tmfour = d_{\tmfour}$; 
					we apply the \ih\ to $\tmfour \partobind {n_1} \tmfour'$ and we have that  $\tmfour \partonotstrat \tmfour'$, or $n_1>0$ and $\tmfour \tostrat \tmthree_1 \partobind {n_1-1} \tmr'$;
					but $\tmfour \partonotstrat \tmfour'$ is impossible because otherwise $\tmfour \tmfive \partonotst \tmfour' \tmfive'$ (as $d_\tmfour \leq d_\tmfive+1$of s);
					therefore, $n_1>0$ and $\tmfour \tostrat \tmthree_1 \partobind {n_1-1} \tmr'$, and so $n>0 $ and  $\tm  = \tmfour\tmfive  \tostrat \tmthree_1\tmfive  \partobindlong {n_1-1+n_2} \tmfour'\tmfive' = \tmtwo$.  
					\item $d=d_{\tmfive}+1 \leq d_{\tmfour}$ and then $d=\deg  {\tmfour \tmfive}\leq \deg \tmfive +1 \leq	  d_{\tmfive}+1 =d $, hence $\deg \tmfive = 
					d_{\tmfive}$; we conclude analogously to thee previous sub-case.	
				\end{enumerate}
			\end{itemize}

			\item \emph{$\beta$ step}:
			\[\AxiomC{$\tmthree \partobind {n_1} \tmthree'$}
			\AxiomC{$\tmfour \partobind {n_2} \tmfour'$}		
			\BinaryInfC{$\tm = (\la\var \tmthree)\tmfour \partobindlong {n_1 + \sizep{\tmthree'}\var \cdot n_2 +1} \tmthree'\isub \var {\tmfour'} = \tmtwo$}
			\DisplayProof\]
			With $n = n_1 + \sizep{\tmthree'}\var \cdot n_2 +1 > 0$. We have $\tm = (\la\var \tmthree)\tmfour \tost \tmthree\isub\var\tmfour$ and by substitutivity of $\partobind n$ (\reflemma{partobind-subs}) $\tmthree\isub\var\tmfour \partobindlong{n_1 + \sizep{\tmthree'}\var \cdot n_2} \tmthree' \isub\var{\tmfour'} = \tmtwo$.
		\end{itemize}
		
		\item\emph{Split}: Exactly as in the head case (\refpropp{macro-head}{split}), using the Indexed Split property for \ll (\refpoint{macro-strat-indexed-split} above).
		\qedhere
	\end{enumerate}
\end{proof}

\newcommand{\tmpp}{\tmfive}

\end{document}